%% file: BFIC.tex
\documentclass[journal]{IEEEtran}

\usepackage{amsmath}
\usepackage{amssymb}
\usepackage{amsthm}
\usepackage{amssymb}
\usepackage{epsfig}
\usepackage{epsf}
\usepackage{subfigure}
\usepackage{graphicx}
\usepackage{cite}
\usepackage{url}
\usepackage[]{authblk}
\newtheorem{claim}{Claim}

\newtheorem{theorem}{Theorem}
\newtheorem{definition}{Definition}

\newtheorem{lemma}{Lemma}

\newtheorem{remark}{Remark}

\newcommand{\Ber}[1]{\mathcal{B}\left(#1\right)} % Bernoulli distribution
\newcommand{\Cdel}[1]{C_{\delta_{#1}}}

\newcommand{\dongningstyle}{true}
\newcommand{\fordiscussion}{false}
\input{./macroset.tex}

\begin{document}

\title{Binary Fading Interference Channel with No CSIT}

% and Delayed Local

\author{Alireza~Vahid,
        Mohammad~Ali~Maddah-Ali,
				A.~Salman~Avestimehr,
        and~Yan~Zhu
        \thanks{A. Vahid is with the School of Electrical and Computer Engineering, Duke University, Durham, NC, USA. Email: {\sffamily alireza.vahid@duke.edu}.}
\thanks{Mohammad~Ali~Maddah-Ali is with the Department of Electrical Engineering, Sharif University of Technology, Tehran, Iran. Email: {\sffamily maddah\_ali@sharif.edu}.}
\thanks{A. S. Avestimehr is with the School of Electrical and Computer Engineering, University of Southern California, Los Angeles, CA, USA. Email: {\sffamily avestimehr@ee.usc.edu}.}
\thanks{Yan~Zhu is with Aerohive Networks Inc., Sunnyvale, CA, USA. Email: {\sffamily zhuyan79@gmail.com}.}
%\thanks{The work of A. S. Avestimehr and A. Vahid is in part supported by NSF Grants CAREER-0953117, CCF-1161720, NETS-1161904, AFOSR Young Investigator Program Award, and ONR award N000141310094.}
\thanks{Preliminary parts of this work was presented in the 2014 International Symposium on Information Theory (ISIT)~\cite{AlirezaNoCSIT}.}
\thanks{Copyright (c) 2014 IEEE.}
}

\maketitle
%%%%%%%%%%%%%%%%%%%%%%%%%%%%%%%%%%%%%%%%%%%%%%%%%%%
\begin{abstract}
We study the capacity region of the two-user Binary Fading (or Erasure) Interference Channel where the transmitters have no knowledge of the channel state information. We develop new inner-bounds and outer-bounds for this problem. We identify three regimes based on the channel parameters: weak, moderate, and strong interference regimes. Interestingly, this is similar to the generalized degrees of freedom  of the two-user Gaussian interference channel where transmitters have perfect channel knowledge. We show that for the weak interference regime, treating interference as erasure is optimal while for the strong interference regime, decoding interference is optimal. For the moderate interference regime, we provide new inner and outer bounds. The inner-bound is based on a modification of the Han-Kobayashi scheme for the erasure channel, enhanced by time-sharing. We study the gap between our inner-bound and our outer-bounds for the moderate interference regime and compare our results to that of the Gaussian interference channel.

Deriving our new outer-bounds has three main steps. We first create a contracted channel that has fewer states compared to the original channel, in order to make the analysis tractable. We then prove the Correlation Lemma that shows an outer-bound on the capacity region of the contracted channel also serves as an outer-bound for the original channel. Finally using the Conditional Entropy Leakage Lemma, we derive our outer-bound on the capacity region of the contracted channel.
\end{abstract}
%%%%%%%%%%%%%%%%%%%%%%%%%%%%%%%%%%%%%%%%%%%%%%%%%%%

% or delayed local
% with no knowledge of the channel state information,
%  with the delayed local knowledge of the channel state information


\section{Introduction}
\label{Sec:Introduction}

\input{Introduction.tex}

%%%%%%%%%%%%%%%%%%%%%%%%%%%%%%%%%%%%%%%%%%%%%%%%%%%
\section{Problem Setting}
\label{Sec:Problem}

\input{Problem.tex}

%%%%%%%%%%%%%%%%%%%%%%%%%%%%%%%%%%%%%%%%%%%%%%%%%%%

\section{Main Results}
\label{Sec:Main}

\input{Main.tex}

%%%%%%%%%%%%%%%%%%%%%%%%%%%%%%%%%%%%%%%%%%%%%%%%%%%

\section{Converse}
\label{Sec:Converse}

In this section, we provide the converse proof of Theorem~\ref{THM:MainOuter} for the no CSIT assumption. We incorporate two lemmas in order to derive the outer-bound that we describe in the following subsection.

% We note that under no CSIT assumption, the transmit signals are independent of the actual channel gain realizations. Thus, without loss of generality, we can assume that ${\sf Rx}_i$ uses the decoding function $\widehat{\hbox{W}}_i = \varphi_i(Y_i^n,G_{ii}^n,G_{\bar{i}i}^n)$. We conclude that any outer-bound for the delayed local CSIT case, serves as an outer-bound for the no CSIT assumption. As a result, it suffices to provide the converse proof of Theorem~\ref{THM:Main} under delayed local CSIT assumption.

\subsection{Key Lemmas}
\label{Sec:KeyLemmas}

\input{KeyLemmasISIT.tex}

\subsection{Deriving the Outer-bounds}
\label{Sec:OuterBounds}

\input{Outer.tex}
\section{Achievability}
\label{Sec:Achievability}

\input{AchievabilityV2.tex}

%%%%%%%%%%%%%%%%%%%%%%%%%%%%%%%%%%%%%%%%%%%%%%%%%%%
%\section{Extension to Delayed Direct-Path CSIT}
%\label{Sec:Delayed}
%
%\input{Delayed.tex}

%%%%%%%%%%%%%%%%%%%%%%%%%%%%%%%%%%%%%%%%%%%%%%%%%%%
\section{Conclusion}
\label{Sec:Conclusion}

\input{Conclusion.tex}

\section*{Acknowledgement}

The work of A. S. Avestimehr and A. Vahid is in part supported by NSF Grants CAREER-0953117, CCF-1161720, NETS-1161904, AFOSR Young Investigator Program Award, ONR award N000141310094, and 2013 Qualcomm Innovation Fellowship.

%%%%%%%%%%%%%%%%%%%%%%%%%%%%%%%%%%%%%%%%%%%%%%%%%%%
%%%%%%%%%%%%%%%%%%%%%%%%%%%%%%%%%%%%%%%%%%%%%%%%%%%

\appendices

%\section{Proof of Lemma~\ref{Lemma:ConditionalLeakagedelayed}}
%\label{Appendix:Leakage}
%
%\input{AppendixLeakageDelayed.tex}
%
%
%\section{Proof of Lemma~\ref{Lemma:ConditionalLeakagedelayed}}
%\label{Appendix:Correlation}
%
%\input{AppendixCorrelation.tex}

\section{Proof of \eqref{eq:TS:common}}
\label{Appendix:TSCOMMON}

\input{TSCOMMON.tex}

\bibliographystyle{ieeetr}
\bibliography{bib_misobc}

\begin{IEEEbiography} [{\includegraphics[width=25mm]{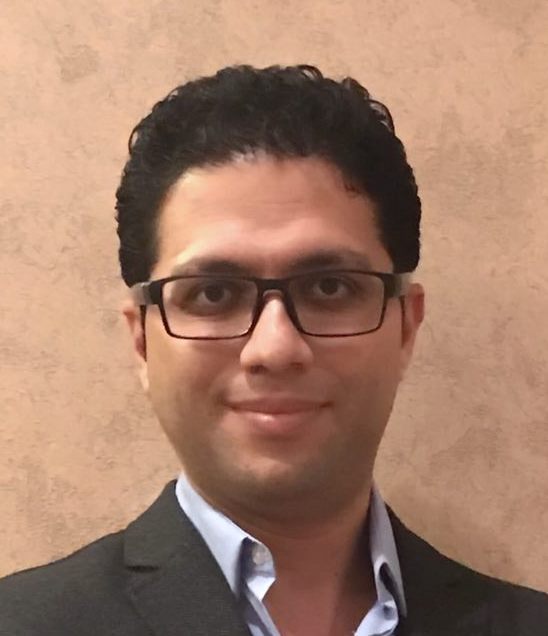}}] {Alireza Vahid} received his Ph.D. and M.Sc. degrees in Electrical and Computer Engineering both from Cornell University, Ithaca, NY, in 2015 and 2012 respectively. He obtained his B.Sc. degree in Electrical Engineering from Sharif University of Technology, Tehran, Iran, in 2009. He is currently a postdoctoral scholar at Information Initiative at Duke University, Durham, NC. His research interests include network information theory, wireless communications, coding theory, and data storage.

Dr. Vahid received the 2015 Outstanding PhD Thesis Research Award at Cornell University. He also received the Director's Ph.D. Teaching Assistant Award in 2010, Jacobs Scholar Fellowship in 2009, and Qualcomm Innovation Fellowship in 2013.
\end{IEEEbiography}

\begin{IEEEbiography} [{\includegraphics[width=25mm]{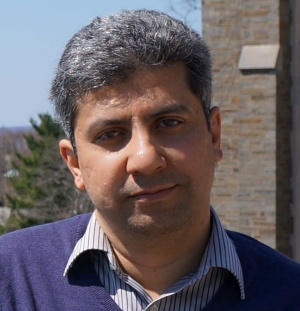}}] {Mohammad Ali Maddah-Ali} (S'03-M'08) received the B.Sc. degree from Isfahan University of Technology, and the M.A.Sc. degree from the University of Tehran, both in electrical engineering. From 2002 to 2007, he was with the Coding and Signal Transmission Laboratory (CST Lab), Department of Electrical and Computer Engineering, University of Waterloo, Canada, working toward the Ph.D. degree. From 2007 to 2008, he worked at the Wireless Technology Laboratories, Nortel Networks, Ottawa, ON, Canada.  From 2008 to 2010, he was a post-doctoral fellow in the Department of Electrical Engineering and Computer Sciences at the University of California at Berkeley. Then, he joined Bell Labs, Holmdel, NJ, as a communication research scientist. Recently, he started working at Sharif University of Technology, as a faculty member. 

Dr. Maddah-Ali is a recipient of NSERC Postdoctoral Fellowship in 2007, a best paper award from IEEE International Conference on Communications (ICC) in 2014, the IEEE Communications Society and IEEE Information Theory Society Joint Paper Award in 2015, and the IEEE Information Theory Society Joint Paper Award in 2016. 
\end{IEEEbiography}

\begin{IEEEbiography} [{\includegraphics[width=25mm]{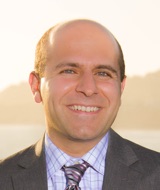}}] {A. Salman Avestimehr} (S'03-M'08-SM’16) is an Associate Professor at the Electrical Engineering Department of University of Southern California. He received his Ph.D. in 2008 and M.S. degree in 2005 in Electrical Engineering and Computer Science, both from the University of California, Berkeley. Prior to that, he obtained his B.S. in Electrical Engineering from Sharif University of Technology in 2003. He was an Assistant Professor at the ECE school of Cornell University from 2009 to 2013. He was also a postdoctoral scholar at the Center for the Mathematics of Information (CMI) at Caltech in 2008. His research interests include information theory, the theory of communications, and their applications.

Dr. Avestimehr has received a number of awards, including the Communications Society and Information Theory Society Joint Paper Award in 2013, the Presidential Early Career Award for Scientists and Engineers (PECASE) in 2011 for "pushing the frontiers of information theory through its extension to complex wireless information networks", the Young Investigator Program (YIP) award from the U. S. Air Force Office of Scientific Research in 2011, the National Science Foundation CAREER award in 2010, and the David J. Sakrison Memorial Prize in 2008. He is currently an Associate Editor for the IEEE Transactions on Information Theory.
\end{IEEEbiography}

\begin{IEEEbiography} [{\includegraphics[width=25mm]{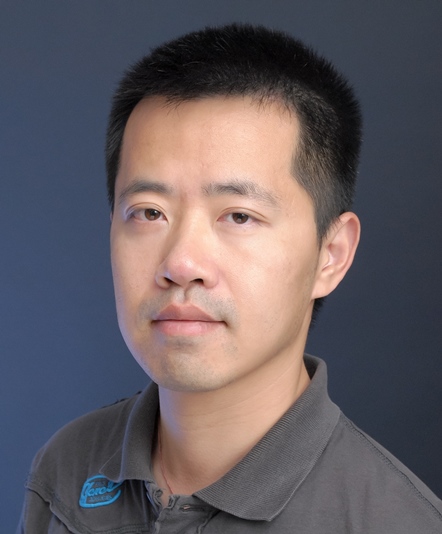}}] {Yan Zhu} received his B.E. and M.S. degrees from Tsinghua University, Beijing, China, in 2002 and 2005, respectively, and his Ph.D. degree from Northwestern University, Evanston, IL, in 2010, all in electrical engineering. From 2010 to 2015, He worked at Broadcom Inc. as a scientist, design staff. He joined Calterah Inc. as a chief system architecturer since 2016. His research interests include wireless communications, information theory, communication network and signal processing. He is a co-recipient of the 2010 IEEE Marconi Prize Paper Award in Wireless Communications (with D. Guo and M. L. Honig).
\end{IEEEbiography}

\end{document}

%% file: macroset.tex
\usepackage{ifthen}
\usepackage{pgffor}

\ifthenelse{\equal{\dongningstyle}{true}}{
        \newcommand{\RVEC}[1]{\boldsymbol{\uppercase{#1}}}
        \newcommand{\RMAT}[1]{\underline{\boldsymbol{\uppercase{#1}}}}
        \newcommand{\RSCA}[1]{\uppercase{#1}}
        \newcommand{\VEC}[1]{\boldsymbol{\lowercase{#1}}}
        \newcommand{\MAT}[1]{\boldsymbol{\underline{\lowercase{#1}}}}
}{
        \newcommand{\RVEC}[1]{\boldsymbol{#1}}
        \newcommand{\RMAT}[1]{\boldsymbol{\mathrm{\uppercase{#1}}}}
        \newcommand{\RSCA}[1]{\uppercase{#1}}
        \newcommand{\VEC}[1]{\boldsymbol{\lowercase{#1}}}
        \newcommand{\MAT}[1]{\boldsymbol{\underline{\lowercase{#1}}}}
}

\ifthenelse{\equal{\fordiscussion}{true}}{
        \newcommand{\NONUM}{ }
}{
        \newcommand{\NONUM}{\nonumber}
}

\foreach \x in {A,...,Z}{%
         \expandafter\xdef\csname rv\x\endcsname{\noexpand\RVEC{\x}}
         \expandafter\xdef\csname rm\x\endcsname{\noexpand\RMAT{\x}}
         \expandafter\xdef\csname r\x\endcsname{\noexpand\RSCA{\x}}
         \expandafter\xdef\csname vv\x\endcsname{\noexpand\VEC{\x}}
         \expandafter\xdef\csname mm\x\endcsname{\noexpand\MAT{\x}}
}

\foreach \x in {a,...,z}{%
         \expandafter\xdef\csname rv\x\endcsname{\noexpand\RVEC{\x}}
         \expandafter\xdef\csname rm\x\endcsname{\noexpand\RMAT{\x}}
         \expandafter\xdef\csname vv\x\endcsname{\noexpand\VEC{\x}}
         \expandafter\xdef\csname mm\x\endcsname{\noexpand\MAT{\x}}
}

\newcommand{\Prob}[1]{\mathsf{P}\left( #1 \right)}

\newcommand{\MI}[1]{\ensuremath{I \left( #1 \right)}}
\newcommand{\CMI}[2]{\ensuremath{I \left(\left. #1 \right| #2 \right)}}
\newcommand{\CMIR}[2]{\ensuremath{I \left( #1 \left| #2 \right. \right)}}
\newcommand{\ENT}[1]{\ensuremath{H\left( #1 \right) }}
\newcommand{\CENT}[2]{\ensuremath{H\left( \left. #1 \right| #2 \right) }}

%% file: Introduction.tex
The two-user Interference Channel (IC) introduced in~\cite{ahlswede1974capacity} is a canonical example to study the impact of interference in communication networks. There exists an extensive body of work on this problem under various assumptions (\emph{e.g.},~\cite{han1981new,sato1977two,ElGamal:it82,van1994some,Carleial,carleial1983outer,Etkin,Suh,AlirezaFB}). In this work, we focus on a specific configuration of this network, named the two-user Binary Fading Interference Channel (BFIC) as depicted in Fig.~\ref{Fig:detIC} in which the channel gains at each time instant are in the binary field according to some Bernoulli distribution. The input-output relation of this channel at time instant $t$ is given by 
\begin{equation} 
\label{eq:receivedsignal}
Y_i[t] = G_{ii}[t] X_i[t] \oplus G_{\bar{i}i}[t] X_{\bar{i}}[t], \quad i = 1, 2,
\end{equation}
where $\bar{i} = 3 - i$, $G_{ii}[t], G_{\bar{i}i}[t] \in \{ 0, 1\}$, and all algebraic operations are in $\mathbb{F}_2$. This model was first introduced in~\cite{vahid2011interference}.

\begin{figure}[ht]
\centering
\includegraphics[height = 3.5cm]{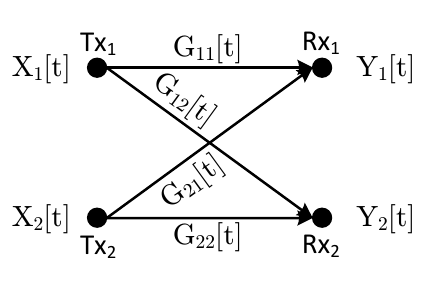}
\caption{Two-user Binary Fading Interference Channel (BFIC).\label{Fig:detIC}}
\end{figure}

The motivation for studying the Binary Fading (or Erasure) Interference Channel is twofold. First as demonstrated in~\cite{vahid2013communication}, it provides a simple yet useful physical layer abstraction for wireless packet networks in which whenever a collision occurs, the receiver can store its received analog signal and utilize it for decoding the packets in the future (for example, by successive interference cancellation techniques). In this context, the binary fading model is motivated by a shadow fading environment in which each link is either ``on'' or ``off'' (according to the shadow fading distribution), and the multiple access (MAC) is modeled such that if two signals are transmitted simultaneously and the links between the corresponding transmitters and the receiver are not in deep fade, then a linear combination of the signals is available to the receiver. The study of the BFIC in~\cite{vahid2013communication} has led to several coding opportunities that can be utilized by the transmitters to exploit the available signal at the receivers for interference management. Moreover, this model allows researchers to focus on other interesting challenges in interference channels such as spatial correlation~\cite{vahid2016does} and locality of channel state knowledge~\cite{vahid2015impact,vahid2016two}.

The second motivation for studying the BFIC is that it can be a first step towards understanding the capacity of the fading interference channels with no knowledge of the channel state information at the transmitters (CSIT). This model was used in~\cite{aggarwal2009ergodic,Guo} to derive the capacity of the one-sided interference channel (also known as Z-Channel). Motivated by the deterministic approach~\cite{ADT10}, a layered erasure broadcast
channel model was introduced in~\cite{tse2012fading} to approximate the capacity of fading broadcast channels.  One can view our Binary Fading model as the model introduced in~\cite{tse2012fading,Guo} with a single layer.

In this work, we consider the two-user BFIC under the no channel state information at the transmitters assumption. In this no CSIT model, the transmitters are only aware of the distributions of the channel gains but not the actual realizations. We develop new inner-bounds and outer-bounds for this problem and we identify three regimes based on the channel parameters: weak, moderate, and strong interference regimes. For the weak and the strong interfrence regimes, we show that the entire capacity region is achieved by applying point-to-point erasure codes with appropriate rates at each transmitter, and using either treat-interference-as-erasure or interference-decoding at each receiver. For the moderate interference regime, we provide new inner and outer bounds. The inner-bound is based on a modification of the Han-Kobayashi scheme for the erasure channel, enhanced by time-sharing. In the moderate interference regime, the inner-bounds and the outer-bounds do not match. We provide some further insights and compare this problem to the two-user static (non-fading) Gaussian interference channel where transmitters have perfect channel knowledge.  

To derive the outer-bound, we incorporate two key lemmas. The first lemma, the Conditional Entropy Leakage Lemma, establishes how much information is leaked from each transmitter to the unintended receiver. The second lemma, the Correlation Lemma, shows that if the channel gains are correlated under a given set of conditions, the capacity region cannot be smaller than the case of independent channel gains. Using the Correlation Lemma, we create a contracted channel that has fewer states as opposed to the original channel and hence, the problem becomes tractable. Then, using the Conditional Entropy Leakage Lemma, we derive an outer-bound on the capacity region of the contracted channel which in turn, serves as an outer-bound for the original channel.

% Later, this problem was studied under a more general delayed {\it local} CSIT assumption and homogenous channel gain distributions in~\cite{FBBudgetISIT,vahid2015value}.

% We use these two lemmas, and we create ``contracted'' channels in which a certain correlation is introduced among the channel gains, and then we derive the outer-bound for these contracted channels. The key idea is that the contracted channels have fewer possible realizations than the original channel which makes the problem tractable.

% For the achievability, we show that the entire capacity region is achieved by applying point-to-point erasure codes with appropriate rates at each transmitter, and using either treat-interference-as-erasure or interference-decoding at each receiver, based on the channel parameters. Thus, message splitting (such as Han-Kobayashi scheme) is not required. We further prove that the achievable region under no CSIT assumption, matches the outer-bound on the capacity region with delayed direct-path CSIT. Thus, the capacity region under the two models would be identical.

% In Section~\ref{Sec:Delayed}, we show that delayed direct-path knowledge of the CSI, does not enlarge the capacity region.

The rest of the paper is organized as follows. In Section~\ref{Sec:Problem}, we formulate our problem. In Section~\ref{Sec:Main}, we present our main results. Section~\ref{Sec:Converse} is dedicated to deriving the outer-bound. We describe our achievability strategy in Section~\ref{Sec:Achievability}. Section~\ref{Sec:Conclusion} concludes the paper and describes future directions.

%% file: Problem.tex
We consider the two-user Binary Fading Interference Channel (BFIC) as illustrated in Fig.~\ref{Fig:detIC}. The channel gain from transmitter ${\sf Tx}_i$ to receiver ${\sf Rx}_j$ at time instant $t$ is denoted by $G_{ij}[t]$, $i,j \in \{1,2\}$. We assume that the channel gains are either $0$ or $1$ (\emph{i.e.} $G_{ij}[t] \in \{0,1\}$), and they are distributed as independent Bernoulli random variables (independent from each other and \emph{over time}). Furthermore, we consider the symmetric setting where 
\begin{align}
G_{ii}[t] \overset{d}\sim \mathcal{B}(p_d) \quad \text{and} \quad G_{i\bar{i}}[t] \overset{d}\sim \mathcal{B}(p_c),
\end{align}
for $0 \leq p_d,p_c \leq 1$, $\bar{i} = 3 - i$, and $i = 1,2$. We define $q_d \overset{\triangle}= 1 - p_d$ and $q_c \overset{\triangle}= 1 - p_c$. 

At each time instant $t$, the transmit signal at ${\sf Tx}_i$ is denoted by $X_i[t] \in \{ 0, 1 \}$, $i = 1, 2$, and the received signal at ${\sf Rx}_i$ is given by
\begin{equation} 
% \label{eq:receivedsignal}
Y_i[t] = G_{ii}[t] X_i[t] \oplus G_{\bar{i}i}[t] X_{\bar{i}}[t], \quad i = 1, 2,
\end{equation}
where all algebraic operations are in $\mathbb{F}_2$. Due to the nature of the channel gains, a total of $16$ channel realizations may occur at any given time instant as given in Table~\ref{Table:16Cases}.

The channel state information (CSI) at time instant $t$ is denoted by the quadruple 
\begin{align}
G[t] = (G_{11}[t], G_{12}[t], G_{21}[t], G_{22}[t]).
\end{align}

\begin{table*}[t]
\caption{All possible channel realizations; solid arrow from transmitter ${\sf Tx}_i$ to receiver ${\sf Rx}_j$ indicates that $G_{ij}[t] = 1$.}
\centering
\begin{tabular}{| c | c | c | c | c | c | c | c |}
\hline
ID		 & ch. realization   & ID		 & ch. realization & ID		 & ch. realization   & ID		 & ch. realization \\ [0.5ex]

\hline

\raisebox{18pt}{$1$}    &    \includegraphics[height = 1.6cm]{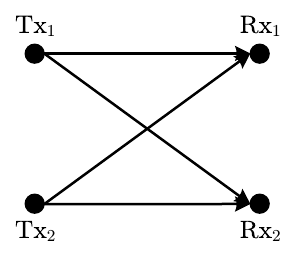}    &  \raisebox{18pt}{$2$}    &    \includegraphics[height = 1.6cm]{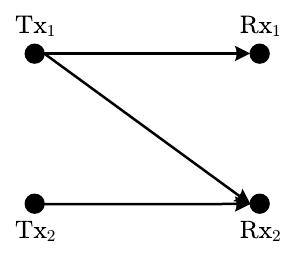}	&	\raisebox{18pt}{$3$}    &    \includegraphics[height = 1.6cm]{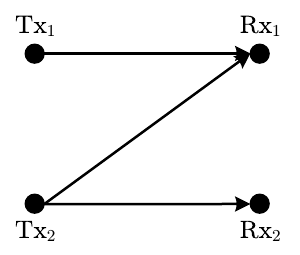}    &  \raisebox{18pt}{$4$}    &    \includegraphics[height = 1.6cm]{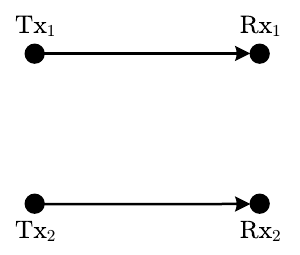} \\

\hline

\raisebox{18pt}{$5$}    &    \includegraphics[height = 1.6cm]{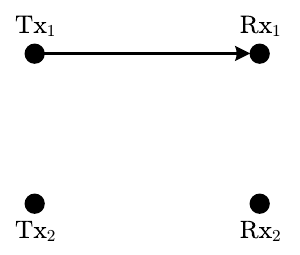}    &  \raisebox{18pt}{$6$}    &    \includegraphics[height = 1.6cm]{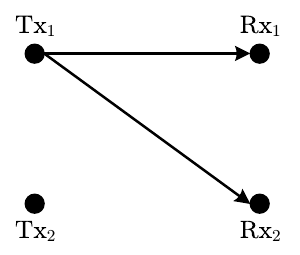}	&	\raisebox{18pt}{$7$}    &    \includegraphics[height = 1.6cm]{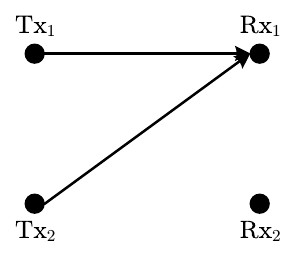}    &  \raisebox{18pt}{$8$}    &    \includegraphics[height = 1.6cm]{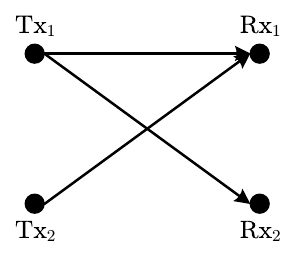} \\

\hline

\raisebox{18pt}{$9$}    &    \includegraphics[height = 1.6cm]{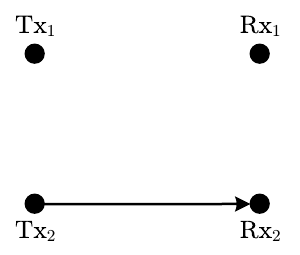}    &  \raisebox{18pt}{$10$}    &    \includegraphics[height = 1.6cm]{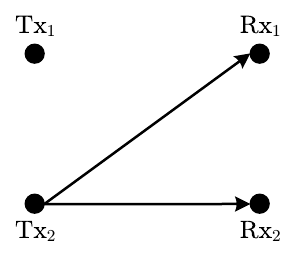}	&	\raisebox{18pt}{$11$}    &    \includegraphics[height = 1.6cm]{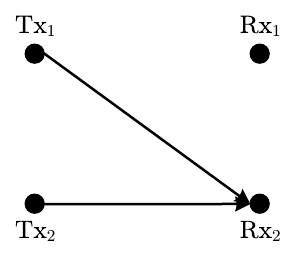}    &  \raisebox{18pt}{$12$}    &    \includegraphics[height = 1.6cm]{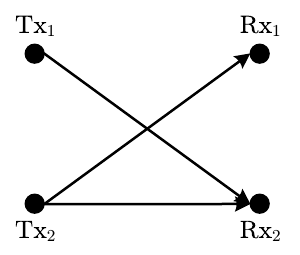} \\

\hline

\raisebox{18pt}{$13$}    &    \includegraphics[height = 1.6cm]{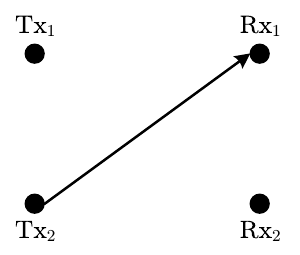}    &  \raisebox{18pt}{$14$}    &    \includegraphics[height = 1.6cm]{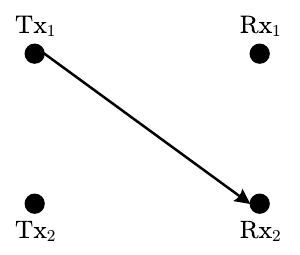}	&	\raisebox{18pt}{$15$}    &    \includegraphics[height = 1.6cm]{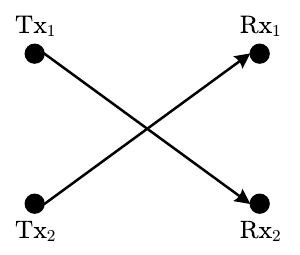}    &  \raisebox{18pt}{$16$}    &    \includegraphics[height = 1.6cm]{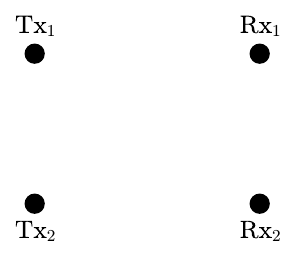} \\

\hline

\end{tabular}
\label{Table:16Cases}
\end{table*}

We use the following notations in this paper. We use capital letters to denote random variables (RVs), \emph{e.g.} $G_{ij}[t]$ is a random variable at time instant $t$, and small letters denote the realizations, \emph{e.g.} $g_{ij}[t]$ is a realization of $G_{ij}[t]$. For a natural number $k$, we set 
\begin{align}
G^{k} = \left[ G[1], G[2], \ldots, G[k] \right]^{\top}.
\end{align}

Finally, we set 
{\small \begin{align}
& G_{ii}^t X_i^t \oplus G_{\bar{i}i}^t X_{\bar{i}}^t \\
&~= \left[ G_{ii}[1] X_i[1] \oplus G_{\bar{i}i}[1] X_{\bar{i}}[1], \ldots, G_{ii}[t] X_i[t] \oplus G_{\bar{i}i}[t] X_{\bar{i}}[t] \right]^{\top}. \nonumber 
\end{align}}

% the receiver only needs to decode the entire message in the very end of the session and thus does not differentiate delayed vs instantaneous CSI.

In this paper, we consider the no CSIT model for the available channel state information at the transmitters. In this model, we assume that transmitters only know the distribution from which the channel gains are drawn, but not the actual realizations of them. Furthermore, we assume that receiver $i$ has instantaneous knowledge of $G_{ii}[t]$ and $G_{\bar{i}i}[t]$ (\emph{i.e.} the incoming links to receiver $i$), $i=1,2$.  

Consider the scenario in which ${\sf Tx}_i$ wishes to reliably communicate message $\hbox{W}_i \in \{ 1,2,\ldots,2^{n R_i}\}$ to ${\sf Rx}_i$ during $n$ uses of the channel, $i = 1,2$. We assume that the messages and the channel gains are {\it mutually} independent and the messages are chosen uniformly. For each transmitter ${\sf Tx}_i$, $i=1,2$, under no CSIT assumption, let message $\hbox{W}_i$ be encoded as $X_i^n$ using the following equation 
\begin{align}
X_i[t] = f_{i,t}\left( \hbox{W}_i \right), \qquad t=1,2,\ldots,n,
\end{align}
where $f_{i,.}\left(.\right)$ is the encoding function at transmitter ${\sf Tx}_i$. 

Receiver ${\sf Rx}_i$ is only interested in decoding $\hbox{W}_i$, and it will decode the message using the decoding function $\widehat{\hbox{W}}_i = \varphi_i\left(Y_i^n, G_{ii}^n, G_{\bar{i}i}^n\right)$. An error occurs when $\widehat{\hbox{W}}_i \neq \hbox{W}_i$. The average probability of decoding error is given by
\begin{equation}
\label{}
\lambda_{i,n} = \mathbb{E}[P[\widehat{\hbox{W}}_i \neq \hbox{W}_i]], \hspace{5mm} i = 1, 2,
\end{equation}
and the expectation is taken with respect to the random choice of the transmitted messages $\hbox{W}_1$ and $\hbox{W}_2$. A rate tuple $(R_1,R_2)$ is said to be achievable, if there exist encoding and decoding functions at the transmitters and the receivers respectively, such that the decoding error probabilities $\lambda_{1,n},\lambda_{2,n}$ go to zero as $n$ goes to infinity. The capacity region $\mathcal{C}\left( p_d, p_c \right)$ is the closure of all achievable rate tuples. In the next section, we present the main results of the paper.

%% file: Main.tex
In this section, we present our main contributions. We first need to define the symmetric sum-rate.
\begin{definition}
For the two-user BFIC with no CSIT, the symmetric sum-rate is defined as
\begin{align}
& R_{\mathrm{sym}}\left( p_d, p_c \right) = R_1 + R_2 \text{~~such~that} \nonumber \\
&~R_1=R_2, ~~ \left( R_1, R_2 \right) \in \mathcal{C}\left( p_d, p_c \right).
\end{align} 
\end{definition}

The following theorems state our main contributions.

\begin{theorem}
\label{THM:MainInner}
For the two-user BFIC with no CSIT the following symmetric sum-rate, $R_{\mathrm{sym}}\left( p_d, p_c \right)$, is achievable:
\begin{itemize}
\item For $0 \leq p_c \leq \frac{p_d}{1+p_d}$: $$2p_dq_c;$$
\item For $\frac{p_d}{1+p_d} \leq p_c \leq p_d$:
$$p_d + p_c - p_d p_c + \frac{p_d - p_c}{2} C_\delta^\ast;$$
where
\begin{align}
  C_\delta^* := \frac{p_dp_c - (p_d - p_c)}{p_dp_c - \frac{p_d - p_c}{2}}.
\end{align}
\item For $p_d \leq p_c \leq 1$:
$$p_d + p_c - p_d p_c.$$
\end{itemize}
\end{theorem}

The following theorem establishes the outer-bound.

\begin{theorem}
\label{THM:MainOuter}
The capacity region of the two-user BFIC with no CSIT, $\mathcal{C}\left( p_d, p_c \right)$, is contained in $\bar{\mathcal{C}}\left( p_d, p_c \right)$ given by:
\begin{itemize}
\item For $0 \leq p_c \leq \frac{p_d}{1+p_d}$: 
\begin{equation}
\label{Eq:OuterRegime1}
% \bar{\mathcal{C}}\left( p_d, p_c \right) =
\left\{ \begin{array}{ll}
\hspace{-1.5mm} \left( R_1, R_2 \right) \left| \parbox[c][4em][c]{0.22\textwidth} {$0 \leq R_i \leq p_d \qquad i=1,2$ \\
$R_i + \beta R_{\bar{i}} \leq \beta p_d + p_c - p_d p_c$} \right. \end{array} \right\}
\end{equation}
where
\begin{align}
\beta = \frac{p_d-p_c}{p_dp_c}.
\end{align}
\item For $\frac{p_d}{1+p_d} \leq p_c \leq p_d$:
\begin{equation}
\label{Eq:OuterRegime2}
% \bar{\mathcal{C}}\left( p_d, p_c \right) =
\left\{ \begin{array}{ll}
\hspace{-1.5mm} \left( R_1, R_2 \right) \left| \parbox[c][4em][c]{0.27\textwidth} {$0 \leq R_i \leq p_d \qquad i=1,2$ \\
$R_i + R_{\bar{i}} \leq 2 p_c$ \\
$R_i + \frac{p_d}{p_c} R_{\bar{i}} \leq \frac{p_d}{p_c} \left( p_d + p_c - p_d p_c \right)$} \right. \end{array} \right\}
\end{equation}
\item For $p_d \leq p_c \leq 1$:
\begin{equation}
\label{Eq:OuterRegime3}
% \bar{\mathcal{C}}\left( p_d, p_c \right) =
\left\{ \begin{array}{ll}
\hspace{-1.5mm} \left( R_1, R_2 \right) \left| \parbox[c][4em][c]{0.22\textwidth} {$0 \leq R_i \leq p_d \qquad i=1,2$ \\
$R_i + R_{\bar{i}} \leq p_d + p_c - p_d p_c$} \right. \end{array} \right\}
\end{equation}
\end{itemize}
\end{theorem}

\begin{figure}[ht]
\centering
\includegraphics[height = 6.5cm]{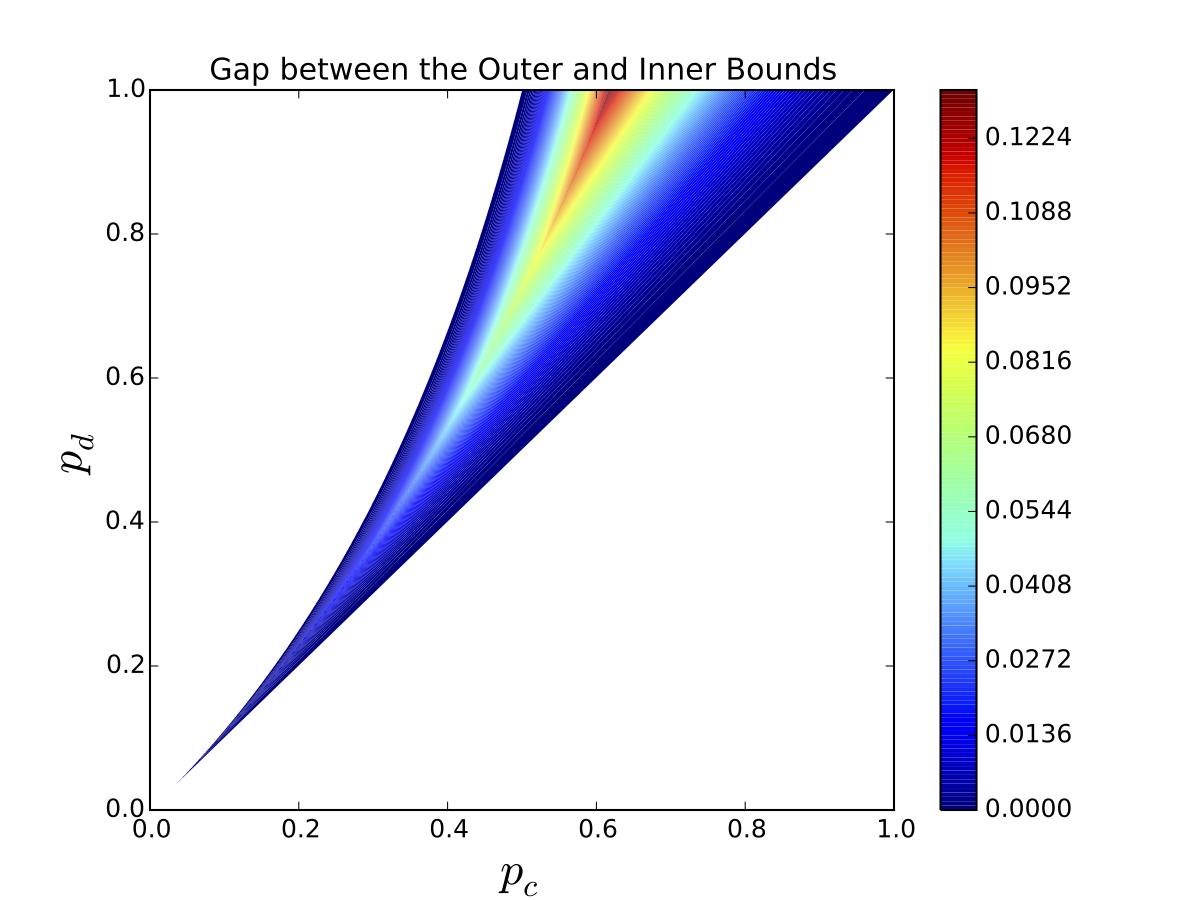}
\caption{The gap between the the inner-bounds and the outer-bounds for the two-user BFIC with no CSIT.\label{Fig:Gap}}
\end{figure}

\begin{figure*}[t]
\centering
\subfigure[]{\includegraphics[height = 4.5 cm]{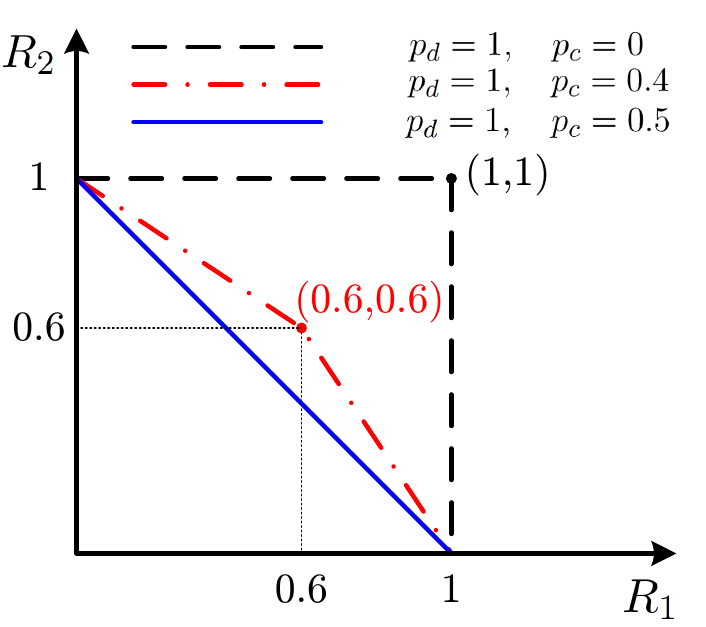}}
\hspace{1.5 in}
\subfigure[]{\includegraphics[height = 4.5 cm]{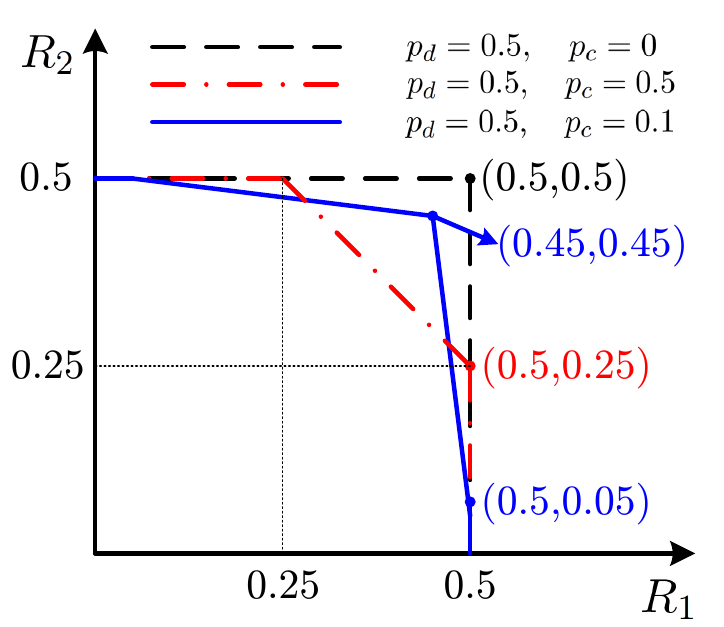}}
%\hspace{0.25 in}
%\subfigure[]{\includegraphics[height = 3.75cm]{FiguresPDF/sumrate.pdf}}
\caption{$(a)$ Capacity region for $p_d = 1$ and $p_c = 0, 0.4, 0.5$; and $(b)$ capacity region  for $p_d = 0.5$ and $p_c = 0, 0.1, 0.5$.\label{Fig:RegionHalf}}
%; and $(c)$ sum-capacity for different values of $p_d$, as $p_c$ moves from $0$ to $1$ 
\end{figure*}

As we show later, the inner-bounds and the outer-bounds for the two-user BFIC with no CSIT match when $0 \leq p_c \leq \frac{p_d}{1+p_d}$ or $p_d \leq p_c \leq 1$ which correspond to the weak and the strong interference regimes. However, for the moderate interference regime, \emph{i.e.} $\frac{p_d}{1+p_d} \leq p_c \leq p_d$, the bounds do not meet, thus the capacity region remains open. The gap between the the inner-bounds and the outer-bounds is plotted in Fig.~\ref{Fig:Gap} where the bounds match in the white region.

For the weak and the strong interference regimes, the capacity region is obtained by applying point-to-point erasure codes with appropriate rates at each transmitter, and using either treat-interference-as-erasure or interference-decoding at each receiver based on the channel parameters. More precisely for $0 \leq p_c < p_d/\left( 1 + p_d \right)$, the capacity region is obtained by treating interference as erasure, while for $p_d \leq p_c \leq 1$, the capacity region is obtained by interference-decoding (\emph{i.e.} the intersection of the capacity regions of the two multiple access channels at receivers). For the moderate interference regime, the inner-bound is based on a modification of the Han-Kobayashi scheme for the erasure channel, enhanced by time-sharing. The detailed proof of the achievability strategy can be found in Section~\ref{Sec:Achievability}.

% We show that to obtain the capacity region, only single layer encoding at the transmitters is sufficient and message splitting (such as Han-Kobayashi scheme) is not required. At the other end, each receiver, based on the values of $p_d$ and $p_c$, either decodes both the intended signal and the interference, or completely ignores the interference and only decodes the intended signal. More precisely, for $p_d/\left( 1 + p_d \right) \leq p_c \leq 1$, we have $\beta = 1$, and the capacity region would be equal to the intersection of capacity regions of the two multiple-access channels (MACs) formed at the receivers and the interference is decoded at each receiver alongside the intended signal; whereas for $0 \leq p_c < p_d/\left( 1 + p_d \right)$, we have $\beta > 1$, and the maximum sum-rate will be outside the capacity region of the MACs formed at the receivers and the interference is not decoded, see Fig.~\ref{Fig:RegionHalf}(b).

%\begin{figure}[ht]
%\centering
%\includegraphics[height = 5.5cm]{FiguresPDF/sumrate1.pdf}
%\caption{\it Sum-capacity of the two-user Binary Fading IC with delayed direct-path CSIT, for different values of $p_d$, as $p_c$ moves from $0$ to $1$.\label{Fig:SumCapacity}}
%\end{figure}

Fig.~\ref{Fig:RegionHalf} illustrates the capacity region for several values of $p_d$ and $p_c$. In Fig.~\ref{Fig:RegionHalf}(a) the capacity region is depicted for $p_d = 1$ and $p_c = 0, 0.4, 0.5$. From Fig.~\ref{Fig:RegionHalf}(b), we conclude that unlike the case in Fig.~\ref{Fig:RegionHalf}(a), decreasing $p_c$ does not necessarily enlarge the capacity region. In fact, for $p_d = 0.5$, we have
\begin{align}
\mathcal{C}\left( 0.5, 0.5 \right) &\not\subseteq \mathcal{C}\left( 0.5, 0.1 \right), \nonumber \\
\mathcal{C}\left( 0.5, 0.1 \right) &\not\subseteq \mathcal{C}\left( 0.5, 0.5 \right).
\end{align}

Using the results of~\cite{vahid2013capacity}, we have plotted the capacity region of the two-user Binary Fading IC for $p_d = p_c = 0.5$ under three different scenarios in Fig.~\ref{fig:illustration}. Under the delayed CSIT, we assume that transmitters become aware of all channel realizations with unit delay and receivers have access to instantaneous CSI, \emph{i.e.}
\begin{align}
X_i[t] = f_{i,t}\left( \hbox{W}_i, G^{t-1} \right), \qquad t=1,2,\ldots,n,
\end{align}
and
\begin{align}
\widehat{\hbox{W}}_i = \varphi_i\left(Y_i^n, G^n \right);
\end{align}
and under instantaneous CSIT model, we assume that at time instant $t$, transmitters and receivers have access to $G^t$. As we can see in Fig.~\ref{fig:illustration}, when transmitters have access to the delayed knowledge of \emph{all} links in the network, the capacity region is strictly larger than the capacity region under the no CSIT assumption. 

\begin{figure}[ht]
\centering
\includegraphics[height = 4.5cm]{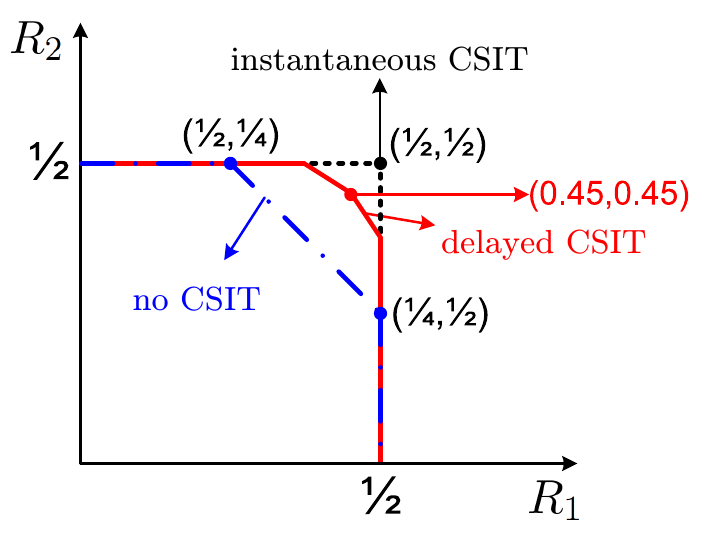}
\caption{\it Capacity region of the two-user Binary Fading IC for $p_d = p_c = 0.5$, with no CSIT, delayed CSIT, and instantaneous CSIT.\label{fig:illustration}}
\end{figure}

To prove Theorem~\ref{THM:MainOuter}, we incorporate two key lemmas as discussed in Section~\ref{Sec:KeyLemmas}. The first step in obtaining the outer-bound is to create a ``contracted'' channel that has fewer states compared to the original channel. Using the Correlation Lemma, we show that an outer-bound on the capacity region of the contracted channel also serves as an outer-bound for the original channel. Finally, using the Conditional Entropy Leakage Lemma, we derive this outer-bound. 

The intuition for the Correlation Lemma was first provided by Sato~\cite{sato1977two}: the capacity region of all interference channels that have the same marginal distributions is the same. We take this intuition and impose a certain spatial correlation among channel gains such that the marginal distributions remain unchanged. For the moderate interference regime, we also incorporate the outer-bounds on the capacity region of the one-sided BFIC with no CSIT~\cite{Guo}.

The rest of the paper is dedicated to the proof of our main contributions. We provide the converse proof in Section~\ref{Sec:Converse}, and we present our achievability strategy in Section~\ref{Sec:Achievability}.

%% file: KeyLemmasISIT.tex
\subsubsection{Entropy Leakage Lemma}

Consider a broadcast channel as depicted in Fig.~\ref{Fig:portion} where a transmitter is connected to two receivers through binary fading channels. Suppose $G_i[t]$ is distributed as i.i.d. Bernoulli random variable (\emph{i.e.} $G_i[t] \overset{d}\sim \mathcal{B}(p_i)$) where $0 \leq p_2 \leq p_1 \leq 1$, $i=1,2$. In this channel the received signals are given as
\begin{align}
Y_i[t] = G_i[t] X[t], \qquad i = 1,2,
\end{align}
where $X[t] \in \{ 0, 1\}$ is the transmit signal at time instant $t$. Furthermore, suppose $G_3[t] \overset{d}\sim \mathcal{B}(p_3)$, $0 \leq p_3 \leq 1$, such that
\begin{align}
\label{eq:NoTwoEqualToOne}
\Pr\left[ G_i[t] = 1, G_j[t] = 1 \right] =  0,~i \neq j,~~i,j \in \{ 1,2,3\}.
%& \Pr\left[ G_1[t] = 1, G_2[t] = 1 \right] =  0, \nonumber \\
%& \Pr\left[ G_1[t] = 1, G_3[t] = 1 \right] =  0, \nonumber \\
%& \Pr\left[ G_2[t] = 1, G_3[t] = 1 \right] =  0,
\end{align} 
% where we assume $p_1+p_2+p_3 \leq 1$. 
We further define
\begin{align}
G_T[t] \overset{\triangle}= \left( G_1[t], G_2[t], G_3[t] \right).
\end{align}
Then, for the channel described above, we have the following lemma.

\begin{figure}[t]
\vspace{-1mm}
\centering
\includegraphics[height = 2.5cm]{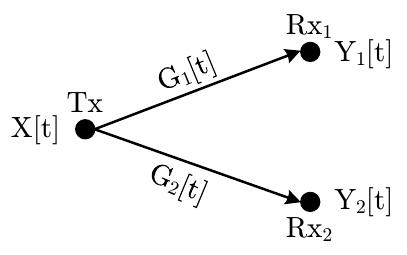}
\caption{A broadcast channel with binary fading links.\label{Fig:portion}}
\vspace{-2mm}
\end{figure}

\begin{lemma}
\label{Lemma:ConditionalLeakage}
[{\it Conditional Entropy Leakage Lemma}] For the channel described above with no CSIT, for $p_1+p_2+p_3 \leq 1$ and \emph{any} input distribution, we have
\begin{align}
\label{eq:leakageeq}
H\left( Y_2^n | G_3^n X^n, G_T^n \right) \geq \frac{p_2}{p_1} H\left( Y_1^n | G_3^n X^n, G_T^n \right).
\end{align}
\end{lemma}

\begin{proof}
Let $G_H[t]$ be distributed as $\mathcal{B}(p_2/p_1)$, and be independent of all other parameters in the network. Let
\begin{align}
Y_H[t] = G_H[t] Y_1[t], \quad t=1,\ldots,n.
\end{align}

It is straightforward to see that $Y_H^t$ is statistically the same as $Y_2^t$ under the no CSIT assumption. For time instant $t$, where $1 \leq t \leq n$, we have
\begin{align}
& H\left( Y_2[t] | Y_2^{t-1}, G_3^n X^n, G_T^n \right) \nonumber \\
& \overset{(a)}= H\left( Y_2[t] | Y_2^{t-1}, G_3^n X^n, G_T^n, G_H^{t-1} \right) \nonumber \\
% & \overset{(b)}= \left( 1 - p_3 \right) H\left( Y_2[t] | Y_2^{t-1}, G_3^n X^n, G_T^n, G_H^{t-1}, G_3[t] = 0 \right) \nonumber \\
& \overset{(b)}= p_2 H\left( X[t] | Y_2^{t-1}, G_3^n X^n, G_T^n, G_H^{t-1}, G_2[t] = 1, G_3[t] = 0 \right) \nonumber \\
& \overset{(c)}= p_2 H\left( X[t] | Y_2^{t-1}, G_3^n X^n, G_T^n, G_H^{t-1}, G_3[t] = 0 \right) \nonumber \\
& \overset{(d)}= p_2 H\left( X[t] | Y_H^{t-1}, G_3^n X^n, G_T^n, G_H^{t-1}, G_3[t] = 0 \right) \nonumber\\
& \overset{(e)}\geq p_2 H\left( X[t] | Y_1^{t-1}, Y_H^{t-1}, G_3^n X^n, G_T^n, G_H^{t-1}, G_3[t] = 0 \right) \nonumber\\
& \overset{(f)}= p_2 H\left( X[t] | Y_1^{t-1}, G_3^n X^n, G_T^n, G_H^{t-1}, G_3[t] = 0 \right) \nonumber\\
& \overset{(g)}= \left( 1 - p_3 \right) \frac{p_2}{p_1} H\left( Y_1[t] | Y_1^{t-1}, G_3^n X^n, G_T^n, G_H^{t-1}, G_3[t] = 0 \right) \nonumber\\
& \overset{(h)}= \frac{p_2}{p_1} H\left( Y_1[t] | Y_1^{t-1}, G_3^n X^n, G_T^n, G_H^{t-1} \right) \nonumber \\
& \overset{(i)}= \frac{p_2}{p_1} H\left( Y_1[t] | Y_1^{t-1}, G_3^n X^n, G_T^n \right),
\end{align}
where $(a)$ holds since $G_H^{t-1}$ is independent of all other parameters in the network; $(b)$ follows from
\begin{align}
\Pr\left[ G_2[t] = 1 | G_3[t] = 0 \right] = \frac{p_2}{1-p_3};
\end{align}
$(c)$ holds since the transmit signal is independent of the channel realizations; $(d)$ follows from the fact that
\begin{align}
&\Pr \left[ X[t], Y_2^{t-1}, G_3^n X^n, G_T^n, G_H^{t-1}, G_3[t] = 0 \right] \nonumber \\
&~= \Pr \left[ X[t], Y_H^{t-1}, G_3^n X^n, G_T^n, G_H^{t-1}, G_3[t] = 0 \right]. 
\end{align}
This equality holds since starting from one side a simple index exchange of $2 \leftrightarrow H$ gives the other side and $X[t]$ is oblivious of the channel realizations;
% $Y_H^{t-1}$ is statistically the same as $Y_2^{t-1}$;
$(e)$ holds since conditioning reduces entropy; $(f)$ is true since
\begin{align}
H\left( Y_H^{t-1} | Y_1^{t-1}, G_3^n X^n, G_T^n, G_H^{t-1} \right) = 0; 
\end{align}
$(g)$ follows from the fact that 
\begin{align}
\Pr\left[ G_1[t] = 1 | G_3[t] = 0 \right] = \frac{p_1}{1-p_3};
\end{align}
$(h)$ is true since $\Pr\left[ G_3[t] = 0 \right] = 1-p_3$; and $(i)$ holds since $G_H^{t-1}$ is independent of all other parameters in the network. Thus, summing all terms for $t = 1,\ldots,n$, we get % the desired result.
\begin{align}
& \sum_{t=1}^{n}{H\left( Y_2[t] | Y_2^{t-1}, G_3^n X^n, G_T^n \right)} \nonumber \\
&~\geq \frac{p_2}{p_1} \sum_{t=1}^{n}{H\left( Y_1[t] | Y_1^{t-1}, G_3^n X^n, G_T^n \right)},
\end{align}
which implies
\begin{align}
H\left( Y_2^n | G_3^n X^n, G_T^n \right) \geq \frac{p_2}{p_1} H\left( Y_1^n | G_3^n X^n, G_T^n \right),
\end{align}
hence, completing the proof.
\end{proof}

\begin{remark}
In~\cite{vahid2013capacity}, we drived the Entropy Leakage Lemma for the case where the transmitter has the CSI with delay. We observe that, with delayed CSIT, the constant on the RHS of (\ref{eq:leakageeq}) would be smaller, meaning that the transmitter can further favor the stronger receiver.
% Lemma~1~\cite{vahid2013capacity} is the Entropy Leakage Lemma for independent channel gains and with delayed CSIT.
\end{remark}

%\begin{remark}
%A similar leakage lemma has also been developed in~\cite{lashgari2013linear}, in the case that there are two distributed transmitters employing linear schemes. In this case, a ``Rank Ratio Inequality,'' is developed that bounds the maximum ratio of the dimensions of received linear-subspaces (at the two receivers) that are createdby distributed transmitters with delayed CSIT.
%\end{remark}

% due to (\ref{eq:NoTwoEqualToOne})

\subsubsection{Correlation Lemma} 

% with fewer states. Let us refer to this channel as the ``contracted'' BFIC, 

Consider again a binary fading interference channel similar to the channel described in Section~\ref{Sec:Problem}, but where channel gains have certain correlation. We denote the channel gain from transmitter ${\sf Tx}_i$ to receiver ${\sf Rx}_j$ at time instant $t$ by $\tilde{G}_{ij}[t]$, $i,j \in \{1,2\}$. We distinguish the RVs in this channel, using $\left( \tilde{.} \right)$ notation (\emph{e.g.}, $\tilde{X}_1[t]$). The input-output relation of this channel at time instant $t$ is given by
\begin{equation} 
\label{Eq:modified}
\tilde{Y}_i[t] = \tilde{G}_{ii}[t] \tilde{X}_i[t] \oplus \tilde{G}_{\bar{i}i}[t] \tilde{X}_{\bar{i}}[t], \quad i = 1, 2.
\end{equation}
% Consider a BFIC similar to what we described in Section~\ref{Sec:Problem}. In this channel that we refer to as ``contracted'' BFIC (since it has fewer states than the original channel), 

We assume that the channel gains are distributed independent over time. However, they can be arbitrary correlated with each other subject to the following constraints.
\begin{align}
\label{eq:constraint}
\Pr & \left( \tilde{G}_{ii}[t] = 1 \right) = p_d, \quad \Pr \left( \tilde{G}_{\bar{i}i}[t] = 1 \right) = p_c, \nonumber \\ 
\Pr & \left( \tilde{G}_{ii}[t] = 1, \tilde{G}_{\bar{i}i}[t] = 1 \right) \nonumber \\
& = \Pr \left( \tilde{G}_{ii}[t] = 1 \right) \Pr \left( \tilde{G}_{\bar{i}i}[t] = 1 \right), \quad i=1,2.
\end{align}
In other words, the channel gains corresponding to incoming links at each receiver are still independent. Similar to the original channel, we assume that the transmitters in this BFIC have no knowledge of the CSI. We have the following result.

% 
%\begin{align}
%& H\left( W_1 | Y_1^n, G^n \right) \leq n \epsilon_n, \nonumber \\
%& H\left( W_2 | Y_2^n, G^n \right) \leq n \epsilon_n.
%\end{align}
%where $\epsilon_n \rightarrow 0$ as $n \rightarrow \infty$. 
% Now, if we plug in $X_1^n$ and $X_2^n$ into the ``contracted'' BFIC, and define then we have the following result.

% From $\tilde{Y}_1^n$ and $\tilde{Y}_2^n$ given in (\ref{Eq:modified}), $\hbox{W}_i$ can be decoded at its ${\sf Rx}_i$ with arbitrary small decoding error probability as $n \rightarrow \infty$, $i=1,2$.

\begin{lemma}
\label{lemma:decodability}
[{\it Correlation Lemma}]  For any BFIC that satisfies the constraints in (\ref{eq:constraint}), we have
\begin{align}
\mathcal{C}\left( p_d, p_c \right) \equiv \tilde{\mathcal{C}}\left( p_d, p_c \right).
\end{align}
\end{lemma}

% The proof is omitted due to space limitations. The main idea is that under no CSIT assumption, without loss of generality, we can assume that ${\sf Rx}_i$ uses the decoding function $\widehat{\hbox{W}}_i = \varphi_i(Y_i^n,G_{ii}^n,G_{\bar{i}i}^n)$ instead of $\widehat{\hbox{W}}_i = \varphi_i(Y_i^n,G^n)$, $i=1,2$.
\begin{proof}
Suppose in the original BFIC messages $\hbox{W}_1$ and $\hbox{W}_2$ are encoded as $X_1^n$ and $X_2^n$ respectively, and each receiver can decode its corresponding message with arbitrary small decoding error probability as $n \rightarrow \infty$. Now, we show that if we use the same transmission scheme in the BFIC that satisfies the constraints in (\ref{eq:constraint}), \emph{i.e.}
\begin{align}
\tilde{X}_i[t] = X_i[t], \quad t = 1,2,\ldots,n, \quad i = 1, 2,
\end{align}
then the receivers in this BFIC can still decode $\hbox{W}_1$ and $\hbox{W}_2$.

In the original two-user binary fading IC as described in Section~\ref{Sec:Problem}, ${\sf Rx}_i$ uses the decoding function $\widehat{\hbox{W}}_i = \varphi_i(Y_i^n,G_{ii}^n,G_{\bar{i}i}^n)$. Therefore, the error event $\widehat{\hbox{W}}_i \neq \hbox{W}_i$, only depends on the choice of $\hbox{W}_i$ and marginal distribution of the channel gains $G_{ii}^n$ and $G_{\bar{i}i}^n$. 

% the decoding function $\widehat{\hbox{W}}_i = \varphi_i(Y_i^n,G^n)$. Note that under no CSIT assumption, $X_i^n$ is independent of the actual channel gain realizations. Hence, without loss of generality, we can assume that ${\sf Rx}_i$ uses

Define 
\begin{align}
\mathcal{E}_{\hbox{W}_1} & = \left\{ \left( \hbox{W}_2, G_{11}^n, G_{21}^n \right) \text{~s.t.~} \widehat{\hbox{W}}_1 \neq \hbox{W}_1 \right\}, \nonumber \\
\tilde{\mathcal{E}}_{\hbox{W}_1} & = \left\{ \left( \hbox{W}_2, \tilde{G}_{11}^n, \tilde{G}_{21}^n \right) \text{~s.t.~} \widehat{\hbox{W}}_1 \neq \hbox{W}_1 \right\},
\end{align}
then, the probability of error is given by
\begin{align}
p_{\text{error}} & = \sum_{w_1}{\Pr \left( \hbox{W}_1 = w_1 \right) \Pr \left( \mathcal{E}_{w_1} \right)} \nonumber \\
& = \frac{1}{2^{nR_1}} \sum_{w_1}{\sum_{\left( w_2, g_{11}^n, g_{21}^n \right) \in \mathcal{E}_{w_1}}{\Pr \left( w_2, g_{11}^n, g_{21}^n \right)}} \\
& \overset{(a)}= \frac{1}{2^{nR_1}} \sum_{w_1}{\sum_{\left( w_2, \tilde{g}_{11}^n, \tilde{g}_{21}^n \right) \in \tilde{\mathcal{E}}_{w_1}}{\Pr \left( w_2, \tilde{g}_{11}^n, \tilde{g}_{21}^n \right)}} = \tilde{p}_{\text{error}}, \nonumber 
\end{align}
where $p_{\text{error}}$ and $\tilde{p}_{\text{error}}$ are the decoding error probability at ${\sf Rx}_1$ in the original and the BFIC satisfying the constraints in (\ref{eq:constraint}) respectively; and $(a)$ holds since according to (\ref{eq:constraint}), the joint distribution of $G_{12}^n$ and $G_{22}^n$ is the same as $\tilde{G}_{12}^n$ and $\tilde{G}_{22}^n$ and the fact that, as mentioned above, the error probability at receiver one only depends on the marginal distribution of these links. Similar argument holds for ${\sf Rx}_2$.
\end{proof}

%% file: Outer.tex
The key to derive the outer-bound is the proper application of the two lemmas introduced in Section~\ref{Sec:KeyLemmas}. More precisely, we need to find a channel that satisfies the constraints in (\ref{eq:constraint}) such that the outer-bound on its capacity region, coincides with the achievable region of the original problem.  We provide the proof for three separate regimes.

\noindent $\bullet$ {\bf Regime~I:} $0 \leq p_c \leq p_d/\left( 1 + p_d \right)$: The derivation of the individual bounds, \emph{e.g.}, $R_1 \leq p_d$, is straightforward and omitted here. We focus on 
\begin{align}
R_1 + \beta R_2 \leq \beta p_d + p_c - p_d p_c,
\end{align}
where 
\begin{align}
\beta = \max \left\{ \frac{p_d-p_c}{p_dp_c}, 1 \right\}.
\end{align}
Due to symmetry, the derivation of the other bound is similar.

\begin{table}[ht]
\caption{The contracted channel for {\bf Regime~I} and {\bf Regime~II}.}
\centering
\begin{tabular}{| c | c || c | c |}
\hline
ID		 & channel realization   & ID		 & channel realization \\ [0.5ex]

\hline

\raisebox{18pt}{$A$}    &    \includegraphics[height = 1.6cm]{FiguresPDF/IC-c1.pdf}    &  \raisebox{18pt}{$B$}    &    \includegraphics[height = 1.6cm]{FiguresPDF/IC-c4.pdf} \\

\hline

                        &    $\Pr \left[ \text{state~A} \right] = p_dp_c$                                                  &                          &    $\Pr \left[ \text{state~B} \right] = p_d-p_c$ \\

\hline
\hline

\raisebox{18pt}{$C$}    &    \includegraphics[height = 1.6cm]{FiguresPDF/IC-c6.pdf}    &  \raisebox{18pt}{$D$}    &    \includegraphics[height = 1.6cm]{FiguresPDF/IC-c10.pdf} \\

\hline

                        &    $\Pr \left[ \text{state~C} \right] = q_dp_c$                                                  &                          &    $\Pr \left[ \text{state~D} \right] = q_dp_c$ \\

\hline

\end{tabular}
\label{Table:ModifiedExample}
\end{table}

% The intuition is that we want $\tilde{G}_{ii}[t] = 1$ whenever $\tilde{G}_{i\bar{i}}[t] = 1$, $i=1,2$. We can simplify the resulting channel to have a contracted channel that has only $5$ states as opposed to the $16$ states for the original channel. 

The first step is to define the appropriate channel that satisfies the constraints in (\ref{eq:constraint}). The idea is to construct a channel such that $\tilde{G}_{ii}[t] = 1$ whenever $\tilde{G}_{i\bar{i}}[t] = 1$, $i=1,2$. We construct such channel with only five states rather than $16$ states and thus, we refer to it as ``contracted'' channel. The five states are denoted by states $A,B,C,D,$ and $E$ with corresponding probabilities $p_dp_c,\left( p_d - p_c \right),q_dp_c,q_dp_c$, and $q_dq_c$. These states are depicted in Table~\ref{Table:ModifiedExample} with the exception of state $E$ which corresponds to the case where all channel gains are $0$. Here, we have
\begin{align}
& \Pr \left( \tilde{G}_{11}[t] = 1 \right) = \sum_{j \in \{ A, B, C\} }{\Pr \left( \text{State~j} \right)} = p_d, \nonumber \\
& \Pr \left( \tilde{G}_{12}[t] = 1 \right) = \sum_{j \in \{ A, C\} }{\Pr \left( \text{State~j} \right)} = p_c, \nonumber \\
& \Pr \left( \tilde{G}_{11}[t] = 1, \tilde{G}_{21}[t] = 1 \right) = \Pr \left( \text{State~A} \right) = p_d p_c,
\end{align}
thus, this channel satisfies the conditions in (\ref{eq:constraint}). From Lemma~\ref{lemma:decodability}, we have $\mathcal{C}\left( p_d, p_c \right) \subseteq \tilde{\mathcal{C}}\left( p_d, p_c \right)$, thus, any outer-bound on the capacity region of the contracted channel, provides an outer-bound on the capacity region of the original channel.

We define 
\begin{align}
\tilde{X}_{1A}[t] \overset{\triangle}= \tilde{X}_1[t] \mathbf{1}_{\left\{ \mathrm{state~}A\mathrm{~occurs~at~time~}t \right\}},
\end{align}
where $\mathbf{1}_{\left\{ \mathrm{state~}A\mathrm{~occurs~at~time~}t \right\}}$ is equal to $1$ when at time $t$ state $A$ occurs. Similarly, we define $\tilde{X}_{1B}[t]$,$\tilde{X}_{1C}[t]$,$\tilde{X}_{1D}[t]$\\,$\tilde{X}_{1E}[t]$,$\tilde{X}_{2A}[t]$,$\tilde{X}_{2B}[t]$,$\tilde{X}_{2C}[t]$,$\tilde{X}_{2D}[t]$ and $\tilde{X}_{2E}[t]$. Therefore, we have
\begin{align}
\tilde{X}_{1}[t] = \tilde{X}_{1A}[t] \oplus \tilde{X}_{1B}[t] \oplus \tilde{X}_{1C}[t] \oplus \tilde{X}_{1D}[t] \oplus \tilde{X}_{1E}[t].
\end{align}

Suppose in the contracted channel, there exist encoders and decoders at transmitters and receivers respectively, such that each receiver can decode its corresponding message with arbitrary small decoding error probability as $\epsilon_n \rightarrow 0$. The derivation of the outer-bound is given in (\ref{eq:converse}) for $\epsilon_n \rightarrow 0$ as $n \rightarrow \infty$; and
\begin{align}
0 \leq p_c \leq p_d/\left( 1 + p_d \right) \Rightarrow \beta = \frac{p_d-p_c}{p_dp_c} > 1.
\end{align}

We have
{\small \begin{align}
\label{eq:converse}
& n \left( \tilde{R}_1 + \beta \tilde{R}_2 - \epsilon_n \right) \nonumber \\
& \overset{(a)}\leq I\left( \tilde{X}_1^n ; \tilde{Y}_1^n | \tilde{G}^n \right) + \beta I\left( \tilde{X}_2^n ; \tilde{Y}_2^n | \tilde{G}^n \right) \nonumber \\
& \overset{(b)}= H\left( \tilde{X}_{2D}^n | \tilde{G}^n \right) + H\left( \tilde{X}_{1B}^n, \tilde{X}_{1C}^n | \tilde{X}_{2D}^n, \tilde{G}^n \right) \nonumber \\
& + H\left( \tilde{X}_{1A}^n \oplus \tilde{X}_{2A}^n | \tilde{X}_{1B}^n, \tilde{X}_{1C}^n, \tilde{X}_{2D}^n, \tilde{G}^n \right) \nonumber\\
& - H\left( \tilde{X}_{2D}^n | \tilde{X}_1^n, \tilde{G}^n \right) - H\left( \tilde{X}_{2A}^n | \tilde{X}_{2D}^n, \tilde{X}_1^n, \tilde{G}^n \right) \nonumber \\
& + \beta H\left( \tilde{X}_{1C}^n | \tilde{G}^n \right) + \beta H\left( \tilde{X}_{2B}^n, \tilde{X}_{2D}^n | \tilde{X}_{1C}^n, \tilde{G}^n \right) \nonumber \\
& + \beta H\left( \tilde{X}_{1A}^n \oplus \tilde{X}_{2A}^n | \tilde{X}_{2B}^n, \tilde{X}_{2D}^n, \tilde{X}_{1C}^n, \tilde{G}^n \right) \nonumber \\
& - \beta H\left( \tilde{X}_{1C}^n | \tilde{X}_2^n, \tilde{G}^n \right) - \beta H\left( \tilde{X}_{1A}^n | \tilde{X}_{1C}^n, \tilde{X}_2^n, \tilde{G}^n \right) \nonumber \\
& \overset{(c)}\leq H\left( \tilde{X}_{2D}^n | \tilde{G}^n \right) + H\left( \tilde{X}_{1B}^n, \tilde{X}_{1C}^n | \tilde{G}^n \right) + H\left( \tilde{X}_{1A}^n \oplus \tilde{X}_{2A}^n | \tilde{G}^n \right) \nonumber \\
& - H\left( \tilde{X}_{2D}^n | \tilde{G}^n \right) - H\left( \tilde{X}_{2A}^n | \tilde{X}_{2D}^n, \tilde{G}^n \right) + \beta H\left( \tilde{X}_{1C}^n | \tilde{G}^n \right) \nonumber \\
& + \beta H\left( \tilde{X}_{2B}^n, \tilde{X}_{2D}^n | \tilde{G}^n \right) + \beta H\left( \tilde{X}_{1A}^n \oplus \tilde{X}_{2A}^n | \tilde{G}^n \right) \nonumber \\
& - \beta H\left( \tilde{X}_{1C}^n | \tilde{G}^n \right) - \beta H\left( \tilde{X}_{1A}^n | \tilde{X}_{1C}^n, \tilde{G}^n \right) \nonumber \\
%& = H\left( \tilde{X}_{1B}^n, \tilde{X}_{1C}^n | \tilde{G}^n \right) + H\left( \tilde{X}_{1A}^n \oplus \tilde{X}_{2A}^n | \tilde{G}^n \right) + \beta H\left( \tilde{X}_{2B}^n, \tilde{X}_{2D}^n | \tilde{G}^n \right) + \beta H\left( \tilde{X}_{1A}^n \oplus \tilde{X}_{2A}^n | \tilde{G}^n \right) \nonumber \\
%& - H\left( \tilde{X}_{2A}^n | \tilde{X}_{2D}^n, \tilde{G}^n \right)  - \beta H\left( \tilde{X}_{1A}^n | \tilde{X}_{1C}^n, \tilde{G}^n \right) \nonumber \\
& \overset{(d)}=  H\left( \tilde{X}_{1C}^n | \tilde{G}^n \right) + H\left( \tilde{X}_{1B}^n | \tilde{X}_{1C}^n, \tilde{G}^n \right) \nonumber \\
& + \left( 1 + \beta \right) H\left( \tilde{X}_{1A}^n \oplus \tilde{X}_{2A}^n | \tilde{G}^n \right)  - \beta H\left( \tilde{X}_{1A}^n | \tilde{X}_{1C}^n, \tilde{G}^n \right) \nonumber \\
& + \beta H\left( \tilde{X}_{2B}^n, \tilde{X}_{2D}^n | \tilde{G}^n \right) - H\left( \tilde{X}_{2A}^n | \tilde{X}_{2D}^n, \tilde{G}^n \right) \nonumber \\
& \overset{(e)}\leq H\left( \tilde{X}_{1C}^n | \tilde{G}^n \right) + \left( 1 + \beta \right) H\left( \tilde{X}_{1A}^n \oplus \tilde{X}_{2A}^n | \tilde{G}^n \right) \nonumber \\
& + \beta H\left( \tilde{X}_{2D}^n | \tilde{G}^n \right) + \beta H\left( \tilde{X}_{2B}^n | \tilde{X}_{2D}^n, \tilde{G}^n\right) - H\left( \tilde{X}_{2A}^n | \tilde{X}_{2D}^n, \tilde{G}^n \right) \nonumber \\
& \overset{(f)}\leq H\left( \tilde{X}_{1C}^n | \tilde{G}^n \right) + \left( 1 + \beta \right) H\left( \tilde{X}_{1A}^n \oplus \tilde{X}_{2A}^n | \tilde{G}^n \right) \nonumber \\
& + \beta H\left( \tilde{X}_{2D}^n | \tilde{G}^n \right) + \left( \beta - \frac{1}{\beta} \right) H\left( \tilde{X}_{2B}^n | \tilde{X}_{2D}^n, \tilde{G}^n\right) \nonumber \\
& \overset{(g)}\leq n q_dp_c + n \left( 1 + \beta \right) p_d p_c + n \beta q_dp_c + n \left( \beta - \frac{1}{\beta} \right) \left( p_d - p_c \right) \nonumber \\
& = n \left( \beta p_d + p_c - p_d p_c \right),
\end{align}}
where $(a)$ follows from Fano's inequality and data processing inequality; $(b)$ follows from the definition of the contracted channel and the chain rule; $(c)$ is true since from Claim~\ref{Claim:Zero} below, we have
\begin{align}
I\left( \tilde{X}_1^n ; \tilde{X}_2^n | \tilde{G}^n \right) = 0,
\end{align}
which results in
\begin{align}
H\left( \tilde{X}_{2A}^n | \tilde{X}_{2D}^n, \tilde{X}_1^n, \tilde{G}^n \right) & = H\left( \tilde{X}_{2A}^n | \tilde{X}_{2D}^n, \tilde{G}^n \right), \nonumber\\
H\left( \tilde{X}_{1A}^n | \tilde{X}_{1C}^n, \tilde{X}_2^n, \tilde{G}^n \right) & = H\left( \tilde{X}_{1A}^n | \tilde{X}_{1C}^n, \tilde{G}^n \right);
\end{align}
and the fact that conditioning reduces entropy; $(d)$ follows from the chain rule; $(e)$ holds since using Lemma~\ref{Lemma:ConditionalLeakage}, we have (use analogy: $\tilde{X}_{1B}^n \leftrightarrow Y_2^n$, $\tilde{X}_{1C}^n \leftrightarrow G_3^nX^n$, and $\tilde{X}_{1A}^n \leftrightarrow Y_1^n$)
\begin{align}
H\left( \tilde{X}_{1B}^n | \tilde{X}_{1C}^n, \tilde{G}^n \right) - \beta H\left( \tilde{X}_{1A}^n | \tilde{X}_{1C}^n, \tilde{G}^n \right) \leq 0;
\end{align}
and $(f)$ follows since from applying Lemma~\ref{Lemma:ConditionalLeakage}, we have (use analogy: $\tilde{X}_{2A}^n \leftrightarrow Y_1^n$, $\tilde{X}_{2D}^n \leftrightarrow G_3^nX^n$, and $\tilde{X}_{1B}^n \leftrightarrow Y_2^n$)
\begin{align}
H\left( \tilde{X}_{2A}^n | \tilde{X}_{2D}^n, \tilde{G}^n \right) \geq \frac{1}{\beta} H\left( \tilde{X}_{2B}^n| \tilde{X}_{2D}^n, \tilde{G}^n \right);
\end{align}
and $(g)$ follows from the fact that entropy of a binary random variable is maximized by i.i.d. Bernoulli distribution with success probability of half. Dividing both sides by $n$ and let $n \rightarrow \infty$, we get the desired result.

\begin{claim}
\label{Claim:Zero}
\begin{align}
I\left( \tilde{X}_1^n ; \tilde{X}_2^n | \tilde{G}^n \right) = 0.
\end{align}
\end{claim}

\begin{proof}
\begin{align}
\label{eq:addX2}
& 0 \leq I\left( \tilde{X}_1^n ; \tilde{X}_2^n | \tilde{G}^n \right) \leq I\left( \tilde{W}_1,\tilde{X}_1^n; \tilde{W}_2,\tilde{X}_2^n|\tilde{G}^n\right) \nonumber \\
& = I\left( \tilde{W}_1 ; \tilde{W}_2|\tilde{G}^n\right) + I\left( \tilde{W}_1; \tilde{X}_2^n| \tilde{W}_2, \tilde{G}^n \right) \nonumber \\
& ~+ I\left( \tilde{X}_1^n; \tilde{W}_2, \tilde{X}_2^n| \tilde{W}_1, \tilde{G}^n \right) = 0.
\end{align}
where the last equality holds since
\begin{align}
I\left( \tilde{W}_1; \tilde{W}_2| \tilde{G}^n \right) = 0,
\end{align}
due to the fact that the messages and the channel gains are mutually independent; 
\begin{align}
I\left( \tilde{W}_1; \tilde{X}_2^n| \tilde{W}_2, \tilde{G}^n \right) = 0,
\end{align}
due to the fact that $\tilde{X}_2^n=f_2(\tilde{W}_2,\tilde{G}^n)$; and 
\begin{align}
I\left( \tilde{X}_1^n; \tilde{W}_2, \tilde{X}_2^n| \tilde{W}_1, \tilde{G}^n \right),
\end{align}
due to the fact that $\tilde{X}_1^n=f_1(\tilde{W}_1,\tilde{G}^n)$.
\end{proof}

\noindent $\bullet$ {\bf Regime~II:} $p_d/\left( 1 + p_d \right) \leq p_c \leq p_d$: In this regime for $i=1,2$, we borrow the outer-bound
\begin{align}
R_i + \frac{p_dp_c}{p_d-p_c+p_dp_c} R_{\bar{i}} \leq p_d + \frac{p_dp_c}{p_d-p_c+p_dp_c} p_cq_d,
\end{align}
from~\cite{Guo}. Thus, we focus on
\begin{align}
R_i + R_{\bar{i}} \leq 2 p_c, \qquad i=1,2.
\end{align}

We use the same contracted channel as for the case of Regime~I. Let $G_S[t]$ be distributed as i.i.d. Bernoulli RV and
\begin{align}
G_S[t] \overset{d}\sim \mathcal{B}(\frac{p_d-p_c}{p_dp_c}), 
\end{align}
and for $i=1,2$, we define
\begin{align}
\label{eq:helper}
& \bar{X}_{iA}[t] \overset{\triangle}= G_S[t] \tilde{X}_{iA}[t], \nonumber \\
& \hat{X}_{iA}[t] \overset{\triangle}= \left( 1 - G_S[t] \right) \tilde{X}_{iA}[t], \nonumber \\
\end{align}

We have
{\small \begin{align}
\label{eq:converse2}
& n \left( \tilde{R}_1 + \tilde{R}_2 - \epsilon_n \right) \overset{(a)}\leq I\left( \tilde{X}_1^n ; \tilde{Y}_1^n | \tilde{G}^n \right) + I\left( \tilde{X}_2^n ; \tilde{Y}_2^n | \tilde{G}^n \right) \nonumber \\
%& = H\left( \tilde{Y}_1^n | \tilde{G}^n \right) - H\left( \tilde{Y}_1^n | \tilde{X}_1^n, \tilde{G}^n \right) + H\left( \tilde{Y}_2^n | \tilde{G}^n \right) - H\left( \tilde{Y}_2^n | \tilde{X}_2^n, \tilde{G}^n \right) \nonumber \\
& \overset{(b)}= H\left( \tilde{X}_{1A}^n \oplus \tilde{X}_{2A}^n, \tilde{X}_{1B}^n, \tilde{X}_{1C}^n, \tilde{X}_{2D}^n | \tilde{G}^n \right) \nonumber \\
& - H\left( \tilde{X}_{2A}^n, \tilde{X}_{2D}^n | \tilde{X}_1^n, \tilde{G}^n \right) \nonumber \\
& + H\left( \tilde{X}_{1A}^n \oplus \tilde{X}_{2A}^n, \tilde{X}_{2B}^n, \tilde{X}_{1C}^n, \tilde{X}_{2D}^n | \tilde{G}^n \right) \nonumber \\
& - H\left( \tilde{X}_{1A}^n, \tilde{X}_{1C}^n | \tilde{X}_2^n, \tilde{G}^n \right) \nonumber \\
& \overset{(c)}= H\left( \tilde{X}_{2D}^n | \tilde{G}^n \right) + H\left( \tilde{X}_{1B}^n, \tilde{X}_{1C}^n | \tilde{X}_{2D}^n, \tilde{G}^n \right) \nonumber \\
& + H\left( \tilde{X}_{1A}^n \oplus \tilde{X}_{2A}^n | \tilde{X}_{1B}^n, \tilde{X}_{1C}^n, \tilde{X}_{2D}^n, \tilde{G}^n \right) \nonumber\\
& - H\left( \tilde{X}_{2D}^n | \tilde{X}_1^n, \tilde{G}^n \right) - H\left( \tilde{X}_{2A}^n | \tilde{X}_1^n, \tilde{X}_{2D}^n, \tilde{G}^n \right) \nonumber \\
& + H\left( \tilde{X}_{1C}^n | \tilde{G}^n \right) + H\left( \tilde{X}_{2B}^n, \tilde{X}_{2D}^n | \tilde{X}_{1C}^n, \tilde{G}^n \right) \nonumber \\
& + H\left( \tilde{X}_{1A}^n \oplus \tilde{X}_{2A}^n | \tilde{X}_{2B}^n, \tilde{X}_{2D}^n, \tilde{X}_{1C}^n, \tilde{G}^n \right) \nonumber \\
& - H\left( \tilde{X}_{1C}^n | \tilde{X}_2^n, \tilde{G}^n \right) - H\left( \tilde{X}_{1A}^n | \tilde{X}_{1C}^n, \tilde{X}_2^n, \tilde{G}^n \right) \nonumber \\
& \overset{(d)}= H\left( \tilde{X}_{2D}^n | \tilde{G}^n \right) + H\left( \tilde{X}_{1B}^n, \tilde{X}_{1C}^n | \tilde{G}^n \right) \nonumber \\
& + H\left( \tilde{X}_{1A}^n \oplus \tilde{X}_{2A}^n  | \tilde{X}_{1B}^n, \tilde{X}_{1C}^n, \tilde{X}_{2D}^n, \tilde{G}^n \right) \nonumber \\
& - H\left( \tilde{X}_{2D}^n | \tilde{G}^n \right) - H\left( \tilde{X}_{2A}^n | \tilde{X}_{2D}^n, \tilde{G}^n \right) \nonumber \\
& + H\left( \tilde{X}_{1C}^n | \tilde{G}^n \right) + H\left( \tilde{X}_{2B}^n, \tilde{X}_{2D}^n | \tilde{G}^n \right) \nonumber \\
& + H\left( \tilde{X}_{1A}^n \oplus \tilde{X}_{2A}^n | \tilde{X}_{2B}^n, \tilde{X}_{2D}^n, \tilde{X}_{1C}^n, \tilde{G}^n \right) \nonumber \\
& - H\left( \tilde{X}_{1C}^n | \tilde{G}^n \right) - H\left( \tilde{X}_{1A}^n | \tilde{X}_{1C}^n, \tilde{G}^n \right) \nonumber 
\end{align}}
where $\epsilon_n \rightarrow 0$ as $n \rightarrow \infty$; $(a)$ follows from Fano's inequality and data processing inequality; $(b)$ holds due to the definition of the contracted channel; $(c)$ follows from the chain rule; $(d)$ is true since from Claim~\ref{Claim:Zero}, we have
\begin{align}
I\left( \tilde{X}_1^n ; \tilde{X}_2^n | \tilde{G}^n \right) = 0.
\end{align}
We use our definition in (\ref{eq:helper}) for the rest of the proof.
{\small \begin{align}
& \overset{(e)}=  H\left( \tilde{X}_{1C}^n | \tilde{G}^n \right) + H\left( \tilde{X}_{1B}^n | \tilde{X}_{1C}^n, \tilde{G}^n \right) \nonumber \\
& + H\left( \tilde{X}_{1A}^n \oplus \tilde{X}_{2A}^n | \tilde{X}_{1B}^n, \tilde{X}_{1C}^n, \tilde{X}_{2D}^n, \tilde{G}^n \right) \nonumber \\
& - H\left( \bar{X}_{1A}^n | \tilde{X}_{1C}^n, \tilde{G}^n \right) - H\left( \hat{X}_{1A}^n | \tilde{X}_{1B}^n, \tilde{X}_{1C}^n, \tilde{G}^n \right) \nonumber \\
& + H\left( \tilde{X}_{2D}^n | \tilde{G}^n \right) + H\left( \tilde{X}_{2B}^n | \tilde{X}_{2D}^n, \tilde{G}^n \right) \nonumber \\
& + H\left( \tilde{X}_{1A}^n \oplus \tilde{X}_{2A}^n | \tilde{X}_{2B}^n, \tilde{X}_{2D}^n, \tilde{X}_{1C}^n, \tilde{G}^n \right) \nonumber \\
& - H\left( \bar{X}_{2A}^n | \tilde{X}_{2D}^n, \tilde{G}^n \right) - H\left( \hat{X}_{2A}^n | \tilde{X}_{2B}^n, \tilde{X}_{2D}^n, \tilde{G}^n \right) \nonumber \\
& \overset{(f)}\leq  H\left( \tilde{X}_{1C}^n | \tilde{G}^n \right) + H\left( \tilde{X}_{1A}^n \oplus \tilde{X}_{2A}^n | \tilde{X}_{1B}^n, \tilde{X}_{1C}^n, \tilde{X}_{2D}^n, \tilde{G}^n \right) \nonumber \\
& - H\left( \hat{X}_{1A}^n | \tilde{X}_{1B}^n, \tilde{X}_{1C}^n, \tilde{G}^n \right) + H\left( \tilde{X}_{2D}^n | \tilde{G}^n \right) \nonumber \\
& + H\left( \tilde{X}_{1A}^n \oplus \tilde{X}_{2A}^n | \tilde{X}_{2B}^n, \tilde{X}_{2D}^n, \tilde{X}_{1C}^n, \tilde{G}^n \right) \nonumber \\ 
& - H\left( \hat{X}_{2A}^n | \tilde{X}_{2B}^n, \tilde{X}_{2D}^n, \tilde{G}^n \right) \nonumber \\
%& =  H\left( \tilde{X}_{1C}^n | \tilde{G}^n \right) + H\left( \bar{X}_{1A}^n \oplus \bar{X}_{2A}^n | \tilde{X}_{1B}^n, \tilde{X}_{1C}^n, \tilde{X}_{2D}^n, \tilde{G}^n \right) + H\left( \hat{X}_{1A}^n \oplus \hat{X}_{2A}^n | \bar{X}_{1A}^n \oplus \bar{X}_{2A}^n, \tilde{X}_{1B}^n, \tilde{X}_{1C}^n, \tilde{X}_{2D}^n, \tilde{G}^n \right) \nonumber \\
%& - H\left( \hat{X}_{1A}^n | \tilde{X}_{1B}^n, \tilde{X}_{1C}^n, \tilde{G}^n \right) + H\left( \tilde{X}_{2D}^n | \tilde{G}^n \right) + H\left( \tilde{X}_{1A}^n \oplus \tilde{X}_{2A}^n | \tilde{X}_{2B}^n, \tilde{X}_{2D}^n, \tilde{X}_{1C}^n, \tilde{G}^n \right) - H\left( \hat{X}_{2A}^n | \tilde{X}_{2B}^n, \tilde{X}_{2D}^n, \tilde{G}^n \right) \nonumber \\
\end{align}}
{\small \begin{align}
& \leq  H\left( \tilde{X}_{1C}^n | \tilde{G}^n \right) + H\left( \bar{X}_{1A}^n \oplus \bar{X}_{2A}^n | \tilde{X}_{1B}^n, \tilde{X}_{1C}^n, \tilde{X}_{2D}^n, \tilde{G}^n \right) \nonumber \\
& + H\left( \hat{X}_{1A}^n, \hat{X}_{2A}^n | \bar{X}_{1A}^n \oplus \bar{X}_{2A}^n, \tilde{X}_{1B}^n, \tilde{X}_{1C}^n, \tilde{X}_{2D}^n, \tilde{G}^n \right) \nonumber \\
& - H\left( \hat{X}_{1A}^n | \tilde{X}_{1B}^n, \tilde{X}_{1C}^n, \tilde{G}^n \right) + H\left( \tilde{X}_{2D}^n | \tilde{G}^n \right) \nonumber \\
& + H\left( \tilde{X}_{1A}^n \oplus \tilde{X}_{2A}^n | \tilde{X}_{2B}^n, \tilde{X}_{2D}^n, \tilde{X}_{1C}^n, \tilde{G}^n \right) \nonumber \\
& - H\left( \hat{X}_{2A}^n | \tilde{X}_{2B}^n, \tilde{X}_{2D}^n, \tilde{G}^n \right) \nonumber \\
& \overset{(g)}\leq  H\left( \tilde{X}_{1C}^n | \tilde{G}^n \right) + H\left( \bar{X}_{1A}^n \oplus \bar{X}_{2A}^n | \tilde{X}_{1B}^n, \tilde{X}_{1C}^n, \tilde{X}_{2D}^n, \tilde{G}^n \right) \nonumber \\
& + H\left( \hat{X}_{2A}^n | \hat{X}_{1A}^n, \bar{X}_{1A}^n \oplus \bar{X}_{2A}^n, \tilde{X}_{1B}^n, \tilde{X}_{1C}^n, \tilde{X}_{2D}^n, \tilde{G}^n \right) \nonumber \\
& + H\left( \tilde{X}_{2D}^n | \tilde{G}^n \right) + H\left( \tilde{X}_{1A}^n \oplus \tilde{X}_{2A}^n | \tilde{X}_{2B}^n, \tilde{X}_{2D}^n, \tilde{X}_{1C}^n, \tilde{G}^n \right) \nonumber \\
& \overset{(h)}\leq n \left( 2 p_c + \epsilon_n \right),
\end{align}}
%  (the reason is that $\bar{X}_{1A}[t]$ and $\tilde{X}_{1B}[t]$ only depend on $\hbox{W}_1$ and $\tilde{G}_{11}^{t-1}$)
where $(e)$ follows from (\ref{eq:helper}) and the construction of the signals, we have
\begin{align}
\Pr \left[ \hat{X}_{1A}^n, \bar{X}_{1A}^n, \tilde{X}_{1C}^n, \tilde{G}^n \right] = \Pr \left[ \hat{X}_{1A}^n, \tilde{X}_{1B}^n, \tilde{X}_{1C}^n, \tilde{G}^n \right],
\end{align}
thus, 
\begin{align}
H\left( \hat{X}_{1A}^n | \bar{X}_{1A}^n, \tilde{X}_{1C}^n, \tilde{G}^n \right) = H\left( \hat{X}_{1A}^n | \tilde{X}_{1B}^n, \tilde{X}_{1C}^n, \tilde{G}^n \right),
\end{align}
and similar statement is true for $\bar{X}_{2A}^n$ and $\tilde{X}_{2B}^n$; $(f)$ holds since from Lemma~\ref{Lemma:ConditionalLeakage}, we have
\begin{align}
H\left( \tilde{X}_{1B}^n | \tilde{X}_{1C}^n, \tilde{G}^n \right) - H\left( \bar{X}_{1A}^n | \tilde{X}_{1C}^n, \tilde{G}^n \right) \leq 0;
\end{align}
$(g)$ holds since 
\begin{align}
&  H\left( \hat{X}_{1A}^n | \bar{X}_{1A}^n \oplus \bar{X}_{2A}^n, \tilde{X}_{1B}^n, \tilde{X}_{1C}^n, \tilde{X}_{2D}^n, \tilde{G}^n \right) \nonumber \\
& - H\left( \hat{X}_{1A}^n | \tilde{X}_{1B}^n, \tilde{X}_{1C}^n, \tilde{G}^n \right) \leq 0;
\end{align}
$(h)$ holds since
\begin{align}
& H\left( \tilde{X}_{1C}^n | \tilde{G}^n \right) \leq n q_dp_c, \nonumber \\
& H\left( \tilde{X}_{2D}^n | \tilde{G}^n \right) \leq n q_dp_c, \nonumber \\
& H\left( \tilde{X}_{1A}^n \oplus \tilde{X}_{2A}^n | \tilde{X}_{2B}^n, \tilde{X}_{2D}^n, \tilde{X}_{1C}^n, \tilde{G}^n \right) \leq n p_dp_c, \nonumber \\
& H\left( \bar{X}_{1A}^n \oplus \bar{X}_{2A}^n | \tilde{X}_{1B}^n, \tilde{X}_{1C}^n, \tilde{X}_{2D}^n, \tilde{G}^n \right) \leq n (p_d - p_c), \nonumber \\
& H\left( \hat{X}_{2A}^n | \hat{X}_{1A}^n, \bar{X}_{1A}^n \oplus \bar{X}_{2A}^n, \tilde{X}_{1B}^n, \tilde{X}_{1C}^n, \tilde{X}_{2D}^n, \tilde{G}^n \right) \nonumber \\
& ~\leq n (p_c + p_dp_c - p_d).
\end{align}

\noindent $\bullet$ {\bf Regime~III:} $p_d \leq p_c \leq 1$: In this regime, again we have $\beta = 1$, and the capacity region would be equal to the intersection of capacity regions of the two MACs formed at the receivers. Under no CSIT assumption, we do not need to create a contracted channel.

The derivation of the outer-bound is easier compared to the other regimes. Basically, ${\sf Rx}_i$ after decoding and removing its corresponding signal, has a stronger channel from ${\sf Tx}_{\bar{i}}$ compared to ${\sf Rx}_{\bar{i}}$, and thus it must be able to decode both $\tilde{\hbox{W}}_i$ and $\tilde{\hbox{W}}_{\bar{i}}$, $i=1,2$. Thus, we have
\begin{align}
R_1 + R_2 \leq p_d + p_c - p_dp_c.
\end{align}

%% file: AchievabilityV2.tex
In this section, we provide the proof of Theorem~\ref{THM:MainInner}. We show that for the weak and the strong interference regimes, the entire capacity region is achieved by applying point-to-point erasure codes with appropriate rates at each transmitter, using either treat-interference-as-erasure or interference-decoding at each receiver, based on the channel parameters. For the moderate interference regime, \emph{i.e.} $\frac{p_d}{1+p_d} \leq p_c \leq p_d$, the inner-bound is based on a modification of the Han-Kobayashi scheme for the erasure channel, enhanced by time-sharing.

\subsection{Weak and Strong Interference Regimes}

Each transmitter applies a point-to-point erasure random code as described in~\cite{Elias}. On the other hand, at each receiver we have two options: $(1)$ interference-decoding; and $(2)$ treat-interference-as-erasure. When a receiver decodes the interference alongside its intended message, the achievable rate region is the capacity region of the multiple-access channel (MAC) formed at that receiver as depicted in Fig.~\ref{Fig:IC-Ach}. The MAC capacity at ${\sf Rx}_1$, is given by
\begin{equation}
\label{eq:inteDecConst}
\left\{ \begin{array}{ll}
\vspace{1mm}  R_1 \leq p_d, & \\
\vspace{1mm}  R_2 \leq p_c, & \\
R_1 + R_2 \leq 1 - q_dq_c. &
\end{array} \right.
\end{equation}

\begin{figure}[ht]
\centering
\includegraphics[height = 3.5cm]{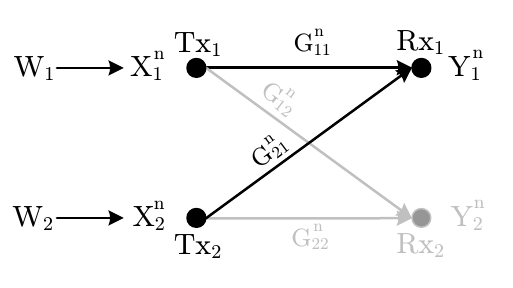}
\caption{\it The multiple-access channel (MAC) formed at ${\sf Rx}_1$.\label{Fig:IC-Ach}}
\end{figure}

As a result, ${\sf Rx}_1$ can decode its message by interference decoding, if $R_1$ and $R_2$ satisfy the constraints in (\ref{eq:inteDecConst}). On the other hand, if ${\sf Rx}_1$ treats interference as erasure, it basically ignores the received signal at time instants where $G_{21}[t] = 1$, and $p_dq_c$ fraction of the time, it receives the transmit signal of ${\sf Tx}_1$. Thus, ${\sf Rx}_1$ can decode its message by treat-interference-as-erasure, if $R_1 \leq p_dq_c$.

Similarly,  ${\sf Rx}_2$ can decode its message by treat-interference-as-erasure, if  $R_2 \leq p_dq_c$. Also, ${\sf Rx}_2$ can decode its message by interference decoding, if $R_1$ and $R_2$ are inside the capacity region of the multiple-access channel (MAC) formed at ${\sf Rx}_2$, i.e.,
\begin{equation}
\label{eq:inteDecConst2}
\left\{ \begin{array}{ll}
\vspace{1mm}  R_1 \leq p_c, & \\
\vspace{1mm}  R_2 \leq p_d, & \\
R_1 + R_2 \leq 1 - q_dq_c. &
\end{array} \right.
\end{equation}

%\overbrace{}^{\text{decode interference}}
%\overbrace{}^{\text{treat int. as erasure at ${\sf Rx}_1$}}

Therefore, the achievable rate region by either treat-interference-as-erasure or interference-decoding at each receiver is the convex hull of (\ref{eq:R}) on top of the next page.
\begin{figure*}[t]
\centering
\begin{equation}
\label{eq:R}
\mathcal{R}= \left \{ (R_1,R_2) \left|   \underbrace{\left( \left\{ \begin{array}{ll}
R_1 \leq p_d & \\
R_2 \leq p_c & \\
R_1 + R_2 \leq 1 - q_dq_c &
\end{array}  \right. \hspace{-2.5mm} \text{or~} R_1 \leq p_dq_c  \right)}_\text{decodability constraint at ${\sf Rx}_1$} \text{ and }  \underbrace{\left(  \left\{ \begin{array}{ll}
R_1 \leq p_c & \\
R_2 \leq p_2 & \\
R_1 + R_2 \leq 1 - q_dq_c &
\end{array} \right. \hspace{-2.5mm} \text{or~} R_2 \leq p_dq_c  \right)}_\text{decodability constraint at ${\sf Rx}_2$}  \right. \right \}.
\end{equation}
\hrule
\end{figure*}

In the remaining of this section, we show that the convex hull of $\mathcal{R}$ matches the outer-bound of Theorem ~\ref{THM:MainOuter} for the weak and the strong interference regimes, \emph{i.e.}
\begin{itemize}
\item For $0 \leq p_c \leq \frac{p_d}{1+p_d}$: 
\begin{equation}
\label{Eq:OuterRegime11}
\bar{\mathcal{C}}\left( p_d, p_c \right) =
\left\{ \begin{array}{ll}
\hspace{-1.5mm} \left( R_1, R_2 \right) \left| \parbox[c][4em][c]{0.22\textwidth} {$0 \leq R_i \leq p_d \qquad i=1,2$ \\
$R_i + \beta R_{\bar{i}} \leq \beta p_d + p_c - p_d p_c$} \right. \end{array} \right\}
\end{equation}
where
\begin{align}
\beta = \frac{p_d-p_c}{p_dp_c}.
\end{align}
\item For $p_d \leq p_c \leq 1$:
\begin{equation}
\label{Eq:OuterRegime33}
\bar{\mathcal{C}}\left( p_d, p_c \right) =
\left\{ \begin{array}{ll}
\hspace{-1.5mm} \left( R_1, R_2 \right) \left| \parbox[c][4em][c]{0.22\textwidth} {$0 \leq R_i \leq p_d \qquad i=1,2$ \\
$R_i + R_{\bar{i}} \leq p_d + p_c - p_d p_c$} \right. \end{array} \right\}
\end{equation}
\end{itemize}

%\begin{equation}
%\label{Eq:CapacityRegionAch}
%\left\{ \begin{array}{ll}
%\hspace{-1.5mm} \left( R_1, R_2 \right) \left| \parbox[c][4em][c]{0.3\textwidth} {$0 \leq R_i \leq p_d \qquad i=1,2$ \\
%$R_i + \beta R_{\bar{i}} \leq \beta p_d + p_c - p_d p_c$} \right. \end{array} \right\}
%\end{equation}
%where
%\begin{align}
%\beta = \max \left\{ \frac{p_d-p_c}{p_dp_c}, 1 \right\}.
%\end{align}

First, it is easy to verify that the region described in (\ref{Eq:OuterRegime11}), is the convex hull of the following corner points:
\begin{align}
\label{eq:cornerP}
& \left( R_1, R_2 \right) = \left( 0, 0 \right), \nonumber \\
& \left( R_1, R_2 \right) = \left( p_d, 0 \right), \nonumber \\
& \left( R_1, R_2 \right) = \left( 0, p_d \right), \nonumber \\
& \left( R_1, R_2 \right) = \left( p_d, q_dp_c \right), \nonumber \\
& \left( R_1, R_2 \right) = \left( q_dp_c, p_d \right), \nonumber \\
& \left( R_1, R_2 \right) = \left( p_dq_c, p_dq_c \right).
\end{align}

% we have plotted $\mathcal{R}_1$ in Fig.~\ref{Fig:IC-AchEx}(a). Then, 

Now, when $0 \leq p_c \leq p_d/\left( 1 + p_d \right)$, the rate region $\mathcal{R}$ (as defined in (\ref{eq:R})) and its convex hull are shown in Fig.~\ref{Fig:IC-AchEx}. As we can note from Fig.~\ref{Fig:IC-AchEx}(b), in this case the convex hull of $\mathcal{R}$ is indeed the convex hull of the corner points in (\ref{eq:cornerP}).

\begin{figure}[ht]
\centering
%\subfigure[]{\includegraphics[height = 4 cm]{FiguresPDF/RegionMACEx.pdf}}
%\hspace{0.5 in}
\subfigure[]{\includegraphics[height = 4.5 cm]{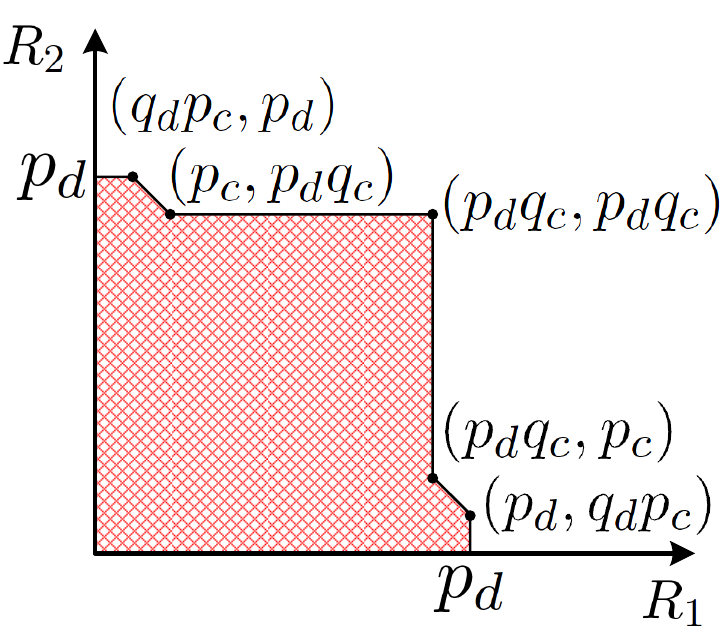}}
\hspace{0.75 in}
\subfigure[]{\includegraphics[height = 4.5 cm]{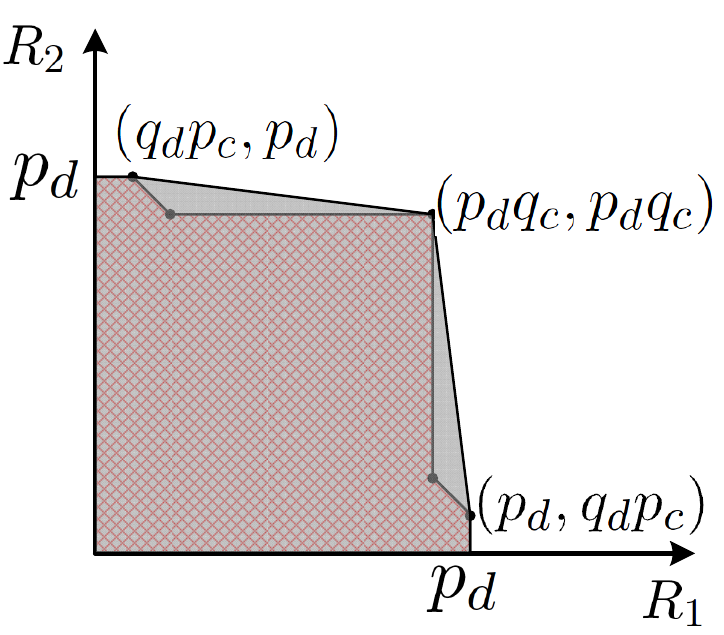}}
\caption{$(a)$ Depiction of the rate region $\mathcal{R}$ for $0 \leq p_c \leq p_d/\left( 1 + p_d \right)$; and $(b)$ its convex hull.\label{Fig:IC-AchEx}}
\end{figure}
%p_d/\left( 1 + p_d \right)

On the other hand, when $p_d \leq p_c \leq 1$, the rate region $\mathcal{R}$ (as defined in (\ref{eq:R})) is depicted in Fig.~\ref{Fig:IC-AchEx2}, which is the convex hull of the first five corner points in (\ref{eq:cornerP}). In this case, the last point in (\ref{eq:cornerP}) is strictly inside the region in Fig.~\ref{Fig:IC-AchEx2}, hence again, $\mathcal{R}$ coincides with the convex hull of the corner points in (\ref{eq:cornerP}), and this completes the proof of Theorem~\ref{THM:MainInner} under no CSIT assumption for the weak and strong interference regimes.

\begin{figure}[ht]
\centering
%\subfigure[]{\includegraphics[height = 4 cm]{FiguresPDF/RegionMACEx2.pdf}}
%\hspace{0.5 in}
%\subfigure[]{}
\includegraphics[height = 4.5 cm]{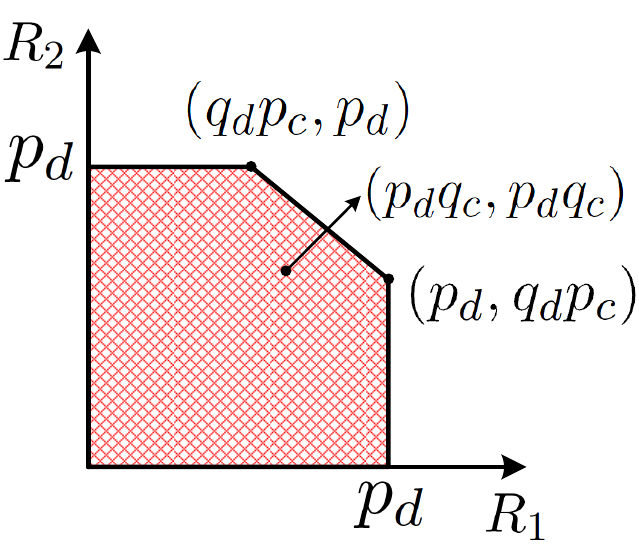}
\caption{Depiction of the rate region $\mathcal{R}$ for $p_d \leq p_c \leq 1$. In this case, corner point $\left( R_1, R_2 \right) = \left( p_dq_c, p_dq_c \right)$ is strictly inside the convex hull of the first five points in (\ref{eq:cornerP}).\label{Fig:IC-AchEx2}}
\end{figure}

\subsection{Moderate Interference Regime}

In the remaining of this section, a Han-Kobayashi~(HK) scheme is proposed and the corresponding sum-rate is investigated. The main focus would be the moderate interference regime, \emph{i.e.} $\frac{p_d}{1+p_d} \leq p_c \leq p_d$. In this regime, neither treating interference as erasure nor completely decoding the interference would be optimal. The result echoes that the sum-capacity might be a similar ``W'' curve as that of the non-fading Gaussian interference channel~\cite{Etkin}. 

To explore the power of HK scheme under this channel model, the codebooks of common and private messages are generated based on rate-splitting, which is originally developed for MAC channels~\cite{grant2001rate}. In a MAC channel, each user can split its own message into two
sub-messages\footnote{In the view-point of a MAC, each user is split into so-called virtual users, which is equivalent to splitting the corresponding message instead.}, and the receiver decodes all messages and sub-messages through an onion-peeling process. That is the receiver
decodes one message by treating others as erasure, and then removes the decoded contribution from the received signal. The receiver
decodes the next message from the remaining signal until no message is left. The key advantage of
rate-splitting scheme for the MAC channel is that the entire capacity region can be achieved with single-user codebooks and no
time-sharing is required. 

We now present the modified HK scheme for the two-user BFIC with no CSIT. Let $\Ber{p}$ denote Bernoulli distribution with probability $p$ taking value 1.  Given $\delta \in [0,1]$, any random
variable $\rX$ with distribution $\Ber{\frac{1}{2}}$ can be split as $\rX = \max(\rX_{c}, \rX_{p})$, where
$\rX_{c} \sim \Ber{\frac{\delta}{2}}$ and $\rX_p \sim \Ber{\frac{1-\delta}{2-\delta}}$. From the view point of the random coding, instead of using $\rX$
to generate a single codebook, one can generate two codebooks according to the distributions of $\rX_c$ and $\rX_p$ respectively, and then the channel input is generated by the $\max$
operator. It is easy to see that $\CENT{\rX}{\rX_c} = \frac{2-\delta}{2}\ENT{\frac{1}{2-\delta}} =:
C_\delta$. Therefore, we have
\begin{align}
\label{eq:basic}
\MI{\rX_{c}; \rX} & =  \ENT{\rX} - \CENT{\rX}{\rX_{c}} = 1 - C_\delta, \NONUM \\
\CMI{\rX_{p}; \rX}{\rX_c} & = \CENT{\rX}{\rX_{c}} =  C_\delta. 
\end{align}
Intuitively, $C_\delta$ can be viewed as the portion of the message carried by codebook generated via
$\rX_p$. Note that $C_\delta$ is continuous and monotonically decreasing with respect to $\delta$. Thus as $\delta$ goes
from $0$ to $1$, the coding scheme continuously changes from $\rX_p$ only scheme ($\delta = 0$) to $\rX_c$
only scheme ($\delta=1$). 

Although time-sharing does not play a key role in achieving the capacity region of the MAC, it does enlarge achievable rate
region of the HK scheme~\cite{te1981new}. For this particular channel, we
consider a Han-Kobayashi scheme with time-sharing as follows. For any $\delta_1 \in [0,1]$, we generate one codebook
according to distribution $\Ber{\frac{\delta_1}{2}}$ for the common message and one random codebook according to distribution $\Ber{ \frac{1-\delta_1}{2-\delta_1}}$ for the private message. We denote these two codebooks by
$\mathcal{C}_c(\delta_1)$ and $\mathcal{C}_p(\delta_1)$ respectively. Similarly for any $\delta_2 \in [0,1]$, we can
generate codebooks $\mathcal{C}_c(\delta_2)$ and $\mathcal{C}_p(\delta_2)$. In
addition, let $\{\rQ[t]\}$ be an i.i.d random sequence with $\Prob{\rQ[t]=1} =
\Prob{\rQ[t]=2} = 1/2$, which generates a particular time-sharing sequence. Before any communication begins,  the
time-sharing sequence is revealed to the transmitters and the receivers. Then the two transmitters communicate their messages as follows. If $\rQ[t] = 1$, user~1
encodes its common and private messages according to codebooks $\mathcal{C}_c(\delta_1)$ and $\mathcal{C}_p(\delta_1)$
respectively, and uses the $\max$ operator to generate
the transmit signal. Meanwhile user 2 does the same thing except that it uses codebooks $\mathcal{C}_c(\delta_2)$ and $\mathcal{C}_p(\delta_2)$. If
$\rQ[t]=2$, the two users switch their codebooks. 

Equivalently, we can state the coding scheme in another way: Given i.i.d. time-sharing random sequence $\{\rQ[t]\}$, the two users encode
their common and private messages independently according to i.i.d. sequences $\{\rX_{1c}[t], \rX_{1p}[t]\}$ and
$\{\rX_{2c}[t], \rX_{2p}[t]\}$. Given $\rQ[t]$, the distributions of the other sequences are defined in
Table~\ref{tbl:1}. 
\begin{table}[h] 
  \centering
  \begin{tabular}{|c|c|c|c|c|}
    \hline  $\rQ[t]$ & $\rX_{1c}[t]$    & $\rX_{1p}[t]$      & $\rX_{2c}[t]$   & $\rX_{2p}[t]$ \\
    \hline   1  & $\Ber{\frac{\delta_1}{2}}$ & $\Ber{1 - \frac{1}{2-\delta_1}}$ & $\Ber{\frac{\delta_2}{2}}$ & $\Ber{\frac{1-\delta_2}{2-\delta_2}}$ \\
    \hline  2 & $\Ber{\frac{\delta_2}{2}}$ & $\Ber{1 - \frac{1}{2-\delta_2}}$ & $\Ber{\frac{\delta_1}{2}}$ & $\Ber{\frac{1-\delta_1}{2-\delta_1}}$  \\ 
    \hline
  \end{tabular}
	\vspace{1mm}
  \caption{Summary of the HK scheme with time-sharing}
  \label{tbl:1}
\end{table}
For fixed $(\delta_1, \delta_2)$, we can compute the corresponding achievable sum-rate, \emph{i.e.}
\begin{align}
R_{sum}(\delta_1, \delta_2) := R_1(\delta_1, \delta_2) + R_2(\delta_1, \delta_2), 
\end{align}
and then maximize over all possible values of
$(\delta_1, \delta_2)$. To determine $R_{sum}(\delta_1, \delta_2)$, it is sufficient to make sure that both common
messages are decodable at both receivers and each private message is decodable at its corresponding receiver. More
specifically, we follow the similar procedure as~\cite[Section~III]{Etkin} except that we explore time-sharing and
optimize the rate splitting.  That is the rates of the common messages are determined by the two virtual compound-MAC
channels at the two receivers, and each private message is decoded only at its corresponding receiver. We have following
theorem:
\begin{theorem}\label{thm:1}
  The following sum-rate is optimal over $\delta_1$, $\delta_2\in [0, 1]$:
  \begin{align} \label{eq:thm1A}
    R_{sum}  = \left\{
    \begin{array}{ll}
      2p_d (1 - p_c) & \text{if }  p_c \leq  \frac{p_d}{1+p_d}\\
      p_d + p_c - p_d p_c + \frac{p_d - p_c}{2} C_\delta^* & \text{if } \frac{p_d}{1+p_d}  < p_c \leq p_d
    \end{array} \right.
  \end{align}
where
\begin{align} \label{eq:thm1B}
  C_\delta^* := \frac{p_dp_c - (p_d - p_c)}{p_dp_c - \frac{p_d - p_c}{2}}.
\end{align}
\end{theorem}
\begin{remark}
  Treating interference as erasure can achieve a sum-rate as large as $2p_d (1 - p_c)$, while the sum-rate of
  interference-decoding scheme is $p_d + p_c - p_d p_c$. In the weak interference regime where
  $p_c \leq \frac{p_d}{1+p_d}$, the proposed HK scheme degrades to treating-interference-as-erasure and achieves
  the sum-capacity. In the moderate interference regime where $\frac{p_d}{1+p_d} < p_c \leq p_d$, partially decoding
  interference can outperform the other two schemes.
\end{remark}

\noindent {\bf Proof of Theorem~\ref{thm:1}:}

If $p_c \leq  \frac{p_d}{1+p_d}$, the sum-rate in Theorem~\ref{thm:1} is the sum-capacity and can be achieved
by treating interference as erasure, which corresponds to the HK scheme with $\delta_1 = \delta_2 = 0$. Thus, we
assume that $\frac{p_d}{1+p_d}  < p_c \leq p_d$. 

Let $\rG[t] = (\rG_{11}[t], \rG_{21}[t])$ and let $R_{ic}(\delta_1, \delta_2)$ and $R_{ip}(\delta_1, \delta_2)$
denote achievable rates of the common and the private messages respectively. Since all codebooks are generated according to
i.i.d. sequences and the channel is memoryless, we will drop all time indices in what follows.

Since at each receiver we decode private message last, we have
\begin{align}
  R_{1p}(\delta_1, \delta_2) = \CMI{\rX_{1p} ; \rY_1}{\rX_{1c}, \rX_{2c}, \rQ, \rG}.  \NONUM
\end{align}

For the common messages, the
achievable rates fall into the intersection of two virtual MACs at the two receivers. Taking symmetric properties of the
channel and the coding scheme into
account, we conclude that any positive rate pair for the common messages satisfying the following conditions is achievable:
\begin{align}
  R_{1c}(\delta_1, \delta_2) + R_{2c}(\delta_1, \delta_2)  & \leq \CMI{\rX_{2c}, \rX_{1c} ; \rY_1}{\rQ, \rG}, \\
  R_{2c}(\delta_1, \delta_2)  & \leq \CMI{\rX_{2c} ; \rY_1 }{\rX_{1c}, \rQ, \rG} \label{eq:mac_single}, \\
  R_{1c}(\delta_1, \delta_2)  & \leq \CMI{\rX_{1c} ; \rY_1 }{\rX_{2c}, \rQ, \rG}. \label{eq:mac_single2}
\end{align}
Since $p_d \geq p_c$, it is easy to see that \eqref{eq:mac_single} implies
\eqref{eq:mac_single2} as far as sum-rate is concerned. 
Therefore, we can achieve the following sum-rate: 
\begin{align}
R_{sum}(\delta_1, \delta_2) & \leq 2M_p(\delta_1, \delta_2) \NONUM \\ 
														&~+ \min\left( M_{c1}(\delta_1, \delta_2), 2M_{c2}(\delta_1, \delta_2) \right), \label{eq:sum:form}
\end{align}
where
\begin{align}
  M_p(\delta_1, \delta_2) & = \CMI{\rX_{1p} ; \rY_1}{\rX_{1c}, \rX_{2c}, \rQ, \rG},  \NONUM \\
  M_{c1}(\delta_1, \delta_2) & = \CMI{\rX_{2c}, \rX_{1c} ; \rY_1 }{\rQ, \rG}, \NONUM \\
  M_{c2}(\delta_1, \delta_2) & = \CMI{\rX_{2c} ; \rY_1 }{\rX_{c1}, \rQ, \rG}. 
\end{align}
The derivation of $M_p(\delta_1, \delta_2)$, $M_{c1}(\delta_1, \delta_2)$, and
$M_{c2}(\delta_1, \delta_2)$ relies on the chain rule and the basic equalities~\eqref{eq:basic} as described below.

Starting from $M_{c1}(\delta_1, \delta_2)$, we have
\begin{align}
& M_{c1}(\delta_1, \delta_2) = \CMI{\rX_{2c}, \rX_{1c} ; \rY_1 }{\rQ, \rG} \NONUM \\
% & = \CENT{\rY_1 }{\rQ, \rG}  - \CENT{\rY_1 }{\rX_{2c} , \rX_{1c}, \rQ, \rG} \NONUM \\
& =  p_d + p_c - p_cp_d  - \CENT{\rY_1}{\rX_{2c} ,
  \rX_{1c}, \rQ, \rG}. \NONUM \\
& = p_d + p_c - p_cp_d - \CENT{\rG_{11} \rX_1 \oplus \rG_{21} \rX_2 }{ \rX_{1c}, \rX_{2c}, \rQ, \rG}  \NONUM \\
&= p_d + p_c - p_cp_d  \NONUM \\
& -  \Big( p_d(1-p_c) \CENT{\rX_1}{ \rX_{1c}, \rQ} + p_c(1-p_d) \CENT{\rX_2}{ \rX_{2c}, \rQ}
   \nonumber \\
& + p_dp_c \CENT{\rX_1 \oplus \rX_2}{\rX_{2c}, \rX_{1c}, \rQ} \Big) \NONUM \\
%&= p_d + p_c - p_cp_d  - \left(  p_d(1-p_c) \frac{\Cdel{1} + \Cdel{2}}{2}  + p_c(1-p_d) \frac{\Cdel{1} + \Cdel{2}}{2} -
%p_d p_c \gamma(\delta_1, \delta_2) \right)  \NONUM \\
& = p_d + p_c - p_cp_d - (p_d + p_c -2p_c p_d) \frac{\Cdel{1} + \Cdel{2}}{2} \NONUM \\
& - p_d p_c \gamma(\delta_1,
    \delta_2) \label{eq:Mc1}
\end{align}
where 
\begin{align}
\gamma(\delta_1, \delta_2) &:= \CENT{\rX_1 \oplus \rX_2}{ \rX_{2c}, \rX_{1c}, \rQ} \NONUM \\
  & = \frac{\delta_1}{2}\Cdel{2} + \frac{\delta_2}{2} \Cdel{1} + \frac{(2-\delta_1)(2-\delta_2)}{4} \ENT{p^*}, 
\end{align}
and $p^* = \frac{2-\delta_2-\delta_1}{(2-\delta_1)(2-\delta_2)}$. In addition, $\gamma(\delta_1,
  \delta_2) $ is upper bounded by $\Cdel{1} +\Cdel{2}$:
  \begin{align}
   \gamma(\delta_1, \delta_2) & \leq  \CENT{\rX_1}{ \rX_{1c}, \rQ} + \CENT{ \rX_2}{ \rX_{2c}, \rQ}  \NONUM \\
															&= \Cdel{1} + \Cdel{2}. \label{eq:gamma:1}
  \end{align}

Next for $M_{c2}(\delta_1, \delta_2)$ using the chain rule and~\eqref{eq:Mc1}, we get
\begin{align}
  M_{c2}(\delta_1, \delta_2) & = \CMI{\rX_{1c}, \rX_{2c} ; \rY_1 }{\rQ, \rG} - \CMI{\rX_{1c} ; \rY_1 }{\rQ, \rG} \NONUM \\
                      & = M_{c1}(\delta_1, \delta_2) - \CMIR{\rX_{1c}; \rG_{11} \rX_1\oplus \rG_{21}\rX_2}{\rQ, \rG}  \NONUM \\
  %& =  M_{c1}(\delta_1, \delta_2) -  p_d(1-p_c) \CMIR{\rX_{1c}; \rX_1}{\rQ, \rG}  \NONUM \\
 & =  M_{c1}(\delta_1, \delta_2) -  p_d(1-p_c) \left[ 1 - \frac{\Cdel{1}+\Cdel{2}}{2}\right] \NONUM \\
 & = p_c - (p_c - p_c p_d) \frac{\Cdel{1} + \Cdel{2}}{2} + p_dp_c \gamma(\delta_1, \delta_2). \label{eq:Mc2}
\end{align}

 Similarly for $M_p(\delta_1, \delta_2)$, we have
 \begin{align}
 & M_p(\delta_1, \delta_2) \NONUM \\
        & = \CMI{\rX_{1p} , \rX_{1c}, \rX_{2c}; \rY_1}{\rQ, \rG} - \CMI{\rX_{1c}, \rX_{2c} ; \rY_1 }{\rQ, \rG}
          \NONUM \\
        & = \CMI{\rX_{1p} , \rX_{1c};\rY_1}{\rQ, \rG} + \CMI{ \rX_{2c} ; \rY_1 }{ \rX_{1p} , \rX_{1c},\rQ, \rG} \NONUM \\
				& - M_{c1}(\delta_1, \delta_2) \NONUM \\
        %& = \CMI{\rX_1;\rY_1}{\rQ, \rG} + \CMI{ \rX_{2c} ; \rG_{21}\rX_2 }{ \rQ, \rG} - M_{c1}(\delta_1, \delta_2) \NONUM \\ 
        & = p_d(1-p_c) + p_c \left[ 1 - \frac{\Cdel{1} + \Cdel{2}}{2}\right] -
          M_{c1}(\delta_1, \delta_2) \NONUM \\
        & = (p_d - 2p_cp_d) \frac{\Cdel{1} + \Cdel{2}}{2} + p_dp_c \gamma(\delta_1, \delta_2) \label{eq:Mp1}.
 \end{align}

Therefore, substituting \eqref{eq:Mc1},  \eqref{eq:Mc2}, and \eqref{eq:Mp1} into~\eqref{eq:sum:form}, we can get an
expression for $R_{sum}(\delta_1, \delta_2)$. If we let $\delta_2 = 1$, then $\Cdel{2} = 0$ and $\gamma(\delta_1,
\delta_2) = \Cdel{1}$.  Furthermore,  
\begin{align}
& R_{sum} \geq \max_{\delta_1\in [0, 1]}R_{sum}(\delta_1, 1) \NONUM \\
& = \max_{\Cdel{1} \in [0,1]} p_d C_{\delta_1} + \min\Big( p_d + p_c - p_d p_c - \frac{p_d + p_c}{2} C_{\delta_1} \, , \Big. \NONUM \\
&\Big. , \, 2p_c - (p_c + p_cp_d)C_{\delta_1}\Big) \NONUM \\
& = \max_{\Cdel{1} \in [0,1]}  \min\Big(p_d + p_c - p_d p_c + \frac{p_d - p_c}{2} C_{\delta_1} \, , \Big. \NONUM \\
&\Big. \, 2p_c - (p_c -p_d + p_cp_d)C_{\delta_1}\Big), \label{eq:sum:t1} 
\end{align}
which achieves the maximum value of $p_d + p_c - p_d p_c + \frac{p_d - p_c}{2} C_\delta^*$ at $C_{\delta_1} = C_{\delta}^*$.
So it is sufficient to show that the converse is true. 

If $M_{c1}(\delta_1, \delta_2) \geq 2M_{c2}(\delta_1, \delta_2)$, we have
\begin{align}
  R_{sum}(\delta_1, \delta_2) & =2M_p(\delta_1, \delta_2) + 2M_{c2}(\delta_1, \delta_2) \NONUM \\
   & = 2p_c - (p_c - p_d + p_d p_c) (\Cdel{1} + \Cdel{2}).
\end{align}
If $M_{c1}(\delta_1, \delta_2) < 2M_{c2}(\delta_1, \delta_2)$, we have
\begin{align}
& R_{sum}(\delta_1, \delta_2)   = 2M_p(\delta_1, \delta_2)  + M_{c1}(\delta_1, \delta_2) \NONUM \\
 & = p_d + p_c - p_d  p_c + (p_d - p_c-2p_dp_c) \frac{\Cdel{1} + \Cdel{2}}{2} \NONUM \\
& + p_d p_c \gamma(\delta_1,
    \delta_2) \NONUM \\
 &\leq p_d + p_c - p_d  p_c + \frac{p_d - p_c}{2} (\Cdel{1} + \Cdel{2}), \label{eq:sum:t2}
\end{align}
where~\eqref{eq:sum:t2} is due to~\eqref{eq:gamma:1}. 

Now, let $\Cdel{3} = \Cdel{1} + \Cdel{2}$. We have
\begin{align}
  R_{sum} &\leq \max_{\Cdel{3} \in [0, 2]} \min\Big( p_d + p_c - p_d  p_c + \frac{p_d - p_c}{2} \Cdel{3} \, , \Big. \NONUM \\
	& \Big. \,  2p_c - (p_c
            - p_d + p_d p_c) \Cdel{3}\Big). \label{eq:sum:t3}
\end{align}
Comparing~\eqref{eq:sum:t3} and~\eqref{eq:sum:t1}, they are in the same form except that~\eqref{eq:sum:t3} is over a
larger domain. However, \eqref{eq:sum:t3} achieves its maximum when $\Cdel{3} = C_\delta^*$, and $C_\delta^* \leq
1$ with the assumption of $\frac{p_d}{1+p_d} < p_c$. Therefore, Theorem~\ref{thm:1} holds.

\noindent {\bf HK Scheme with Time-Sharing:}

Let $\rQ$ be a random variable taking values in $\{1, 2\}$ with $\Prob{\rQ = i} = \lambda_i$ for $i = 1, 2$.
Coding scheme is generated based on

\begin{table}[h]
  \centering
  \begin{tabular}{|c|c|c|c|c|}
    \hline  $\rQ$ & $\rX_{1c}$    & $\rX_{1p}$      & $\rX_{2c}$   & $\rX_{2p}$ \\
    \hline   1  & $\Ber{\frac{\delta_1}{2}}$ & $\Ber{\frac{1-\delta_1}{2-\delta_1}}$ & $\Ber{\frac{1}{2}}$ & $\Ber{0}$ \\
    \hline  2 & $\Ber{\frac{1}{2}}$ & $\Ber{0}$ & $\Ber{\frac{\delta_2}{2}}$ & $\Ber{\frac{1-\delta_2}{2-\delta_2}}$  \\ 
    \hline
  \end{tabular}
\end{table}

The evaluation of achievable rate is similar to the previous case but it breaks the symmetric property in general. 
Define following short-hand notation:
\begin{align}
  C_{\delta_i} :=  \frac{2 - \delta_i}{2} \ENT{\frac{1}{2-\delta_i}} \qquad i = 1, 2
\end{align}

Any rate pair satisfying the following conditions is achievable for the common messages (see Appendix~\ref{Appendix:TSCOMMON} for the details). 
\begin{subequations} \label{eq:TS:common}
  \begin{align}
    R_{1c} &\leq p_c - \lambda_1 \Cdel{1} p_c - \lambda_2\Cdel{2} p_c p_d,  \\
    R_{2c} & \leq p_c - \lambda_2 \Cdel{2} p_c - \lambda_1 \Cdel{1} p_d p_c,  \\
    R_{1c} + R_{2c} & \leq   p_d + p_c - p_d p_c  - \\
		&~\max(\lambda_1 \Cdel{1} p_d + \lambda_2 \Cdel{2} p_c, \lambda_1\Cdel{1} p_c + \lambda_2 \Cdel{2} p_d  ) \NONUM
  \end{align}
\end{subequations}

Rates of the private messages are given by
\begin{align}
  R_{1p} & \leq \CMI{\rX_{1p} ; \rY_1}{\rX_{1c}, \rX_{2c}, \rQ, \rQ, \rG}  = \lambda_1 p_d \Cdel{1}, \\
  R_{2p} & \leq \CMI{\rX_{2p} ; \rY_2}{\rX_{1c}, \rX_{2c}, \rQ, \rQ, \rG}  = \lambda_2 p_d \Cdel{2}. 
\end{align}

Finally, we evaluate sum-rate with $\lambda_i=1/2$ and $\Cdel{1} = \Cdel{2} = C_{\delta}$:
\begin{align}
  & R_{sum}(C_{\delta}) = p_d C_{\delta} + \min( p_d + p_c - p_d p_c - \frac{p_d + p_c}{2} C_{\delta} \, \NONUM \\
	& , \, 2p_c - (p_c +
  p_cp_d)C_{\delta}) 
\end{align}

\noindent {\bf Some numeric results:}

For $p_d = 1$ and for $p_c \in [1/2, 1]$, we plotted the HK rate (optimized over $\delta$) in Fig.~\ref{Fig:Optimized}. 

\begin{figure}[h]
  \centering
  \includegraphics[width=\columnwidth]{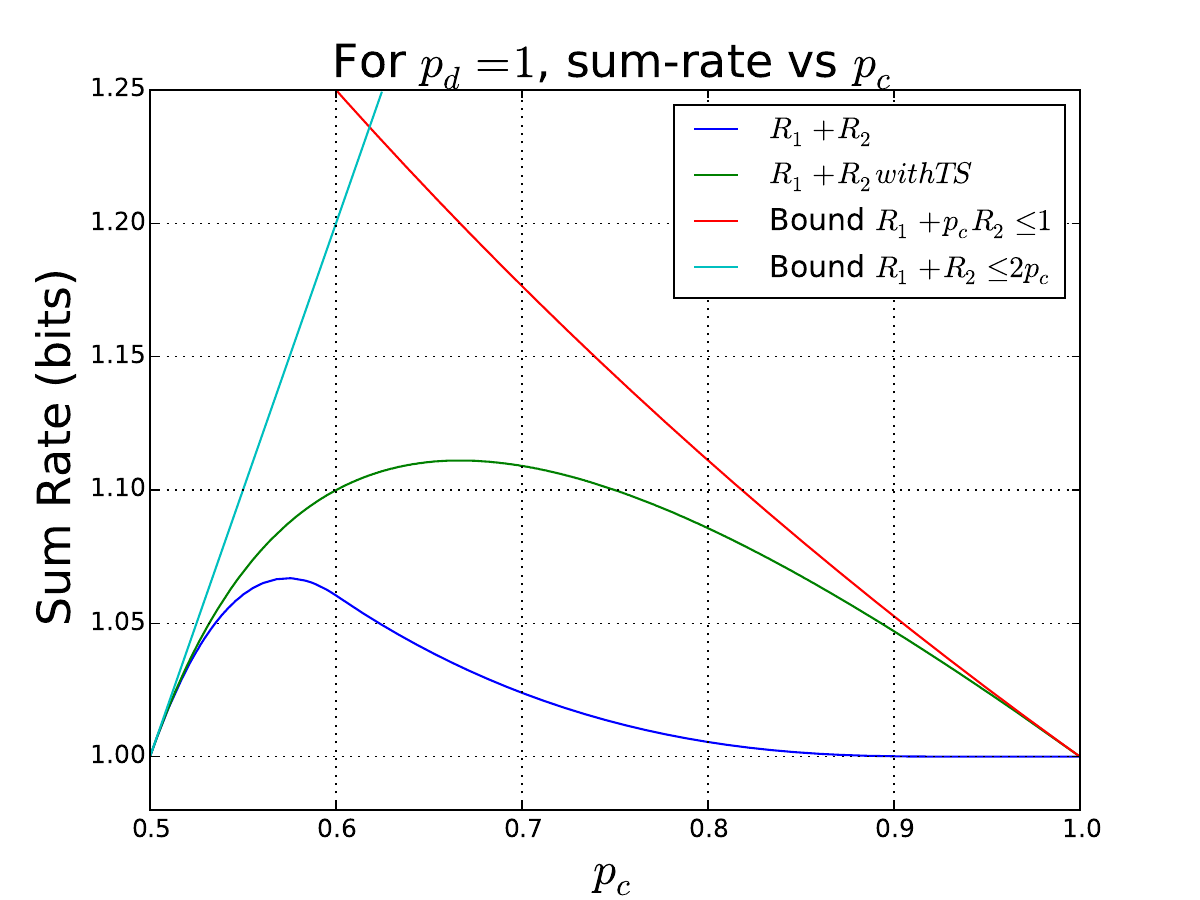}
  \label{Fig:Optimized}
  \caption{The optimized Han-Kobayashi scheme for $p_d = 1$ and for $p_c \in [1/2, 1]$.}
\end{figure}

%% file: Conclusion.tex
We studied the capacity region of the two-user Binary Fading Interference Channel with no CSIT. We showed that under the weak and the moderate interference regimes, the entire capacity region is achieved by applying point-to-point erasure codes with appropriate rates at each transmitter, using either treat-interference-as-erasure or interference-decoding at each receiver, based on the channel parameters. For the moderate interference regime, we devised a modified Han-Kobayashi scheme suited for discrete memoryless channels enhanced by time-sharing. Our outer-bounds rely on two key lemmas, namely the Correlation Lemma and the Entropy Leakage Lemma.

%% file: TSCOMMON.tex
\noindent For the virtual MAC at receiver 1:
\begin{align}
& \CMI{\rX_{1c}, \rX_{2c}; \rY_1}{\rQ, \rG}  \nonumber \\
& = \CENT{\rY_1}{\rQ, \rG} - \CENT{\rY_1}{\rX_{1c}, \rX_{2c}, \rQ, \rG} \NONUM \\
& = p_d + p_c - p_d p_c - \lambda_1 \CENT{\rY_1}{\rX_{1c}, \rX_{2c}, \rQ = 1, \rG} \NONUM \\
& ~+ \lambda_2 \CENT{\rY_1}{\rX_{1c}, \rX_{2c}, \rQ=2, \rG} \NONUM \\
& = p_d + p_c - p_d p_c - \lambda_1 \Cdel{1} p_d - \lambda_2 \Cdel{2} p_c 
\end{align}
\begin{align}
& \CMI{\rX_{2c}; \rY_1}{\rX_{1c}, \rQ, \rG} \nonumber \\
& = \CMI{\rX_{1c}, \rX_{2c}; \rY_1}{ \rQ, \rG}  - \CMI{\rX_{1c}; \rY_1}{ \rQ, \rG} \NONUM \\
& = \CMI{\rX_{1c}, \rX_{2c}; \rY_1}{ \rQ, \rG}  - \lambda_1 \CMI{\rX_{1c}; \rY_1}{ \rQ=1, \rG} \NONUM \\
&~+ \lambda_2 \CMI{\rX_{1c}; \rY_1}{ \rQ=2, \rG} \NONUM \\
& = \CMI{\rX_{1c}, \rX_{2c}; \rY_1}{ \rQ, \rG} - \lambda_1 p_d(1-p_c)(1-\Cdel{1}) \NONUM\\
&~- \lambda_2 p_d (1-p_c) \NONUM \\
& = p_d + p_c - p_d p_c - \lambda_1 \Cdel{1} p_d  - \lambda_2 \Cdel{2} p_c \NONUM\\
&~- \lambda_1 (1-\Cdel{1}) p_d(1-p_c) - \lambda_2 p_d (1-p_c) \NONUM \\
& = p_c - \lambda_1 \Cdel{1}  p_d  - \lambda_2 \Cdel{2}  p_c + \lambda_1 \Cdel{1}  p_d(1-p_c) \NONUM \\
& = p_c - \lambda_2 \Cdel{2} p_c - \lambda_1 \Cdel{1} p_d p_c 
\end{align}
\begin{align}
& \CMI{\rX_{1c}; \rY_1}{\rX_{2c}, \rQ, \rG} \nonumber \\
& = \CMI{\rX_{1c}, \rX_{2c} ; \rY_1}{\rQ, \rG} - \CMI{\rX_{2c}; \rY_1}{\rQ, \rG} \NONUM \\
& = \CMI{\rX_{1c}, \rX_{2c} ; \rY_1}{\rQ, \rG} - \lambda_1 \CMI{\rX_{2c}; \rY_1}{\rQ=1, \rG} \NONUM \\
&~- \lambda_2 \CMI{\rX_{2c}; \rY_1}{\rQ=2, \rG} \NONUM \\
& = \CMI{\rX_{1c}, \rX_{2c} ; \rY_1}{\rQ, \rG} - \lambda_1 p_c(1-p_d) \NONUM \\
&~- \lambda_2 p_c(1-p_d)(1 - \Cdel{2}) \NONUM \\
& = p_d + p_c - p_d p_c - \lambda_1 \Cdel{1} p_d - \lambda_2 \Cdel{2} p_c  \NONUM \\
&~- \lambda_1 p_c(1-p_d) - \lambda_2 p_c(1-p_d)(1 - \Cdel{2}) \NONUM \\
& = p_d - \lambda_1 \Cdel{1} p_d  - \lambda_2 \Cdel{2}  p_c + \lambda_2 \Cdel{2} p_c(1-p_d) \NONUM \\
& = p_d - \lambda_1 \Cdel{1}  p_d  - \lambda_2 \Cdel{2} p_cp_d 
\end{align}

\begin{figure}[ht]
  \centering
  \subfigure[]{\includegraphics[width=.5\textwidth]{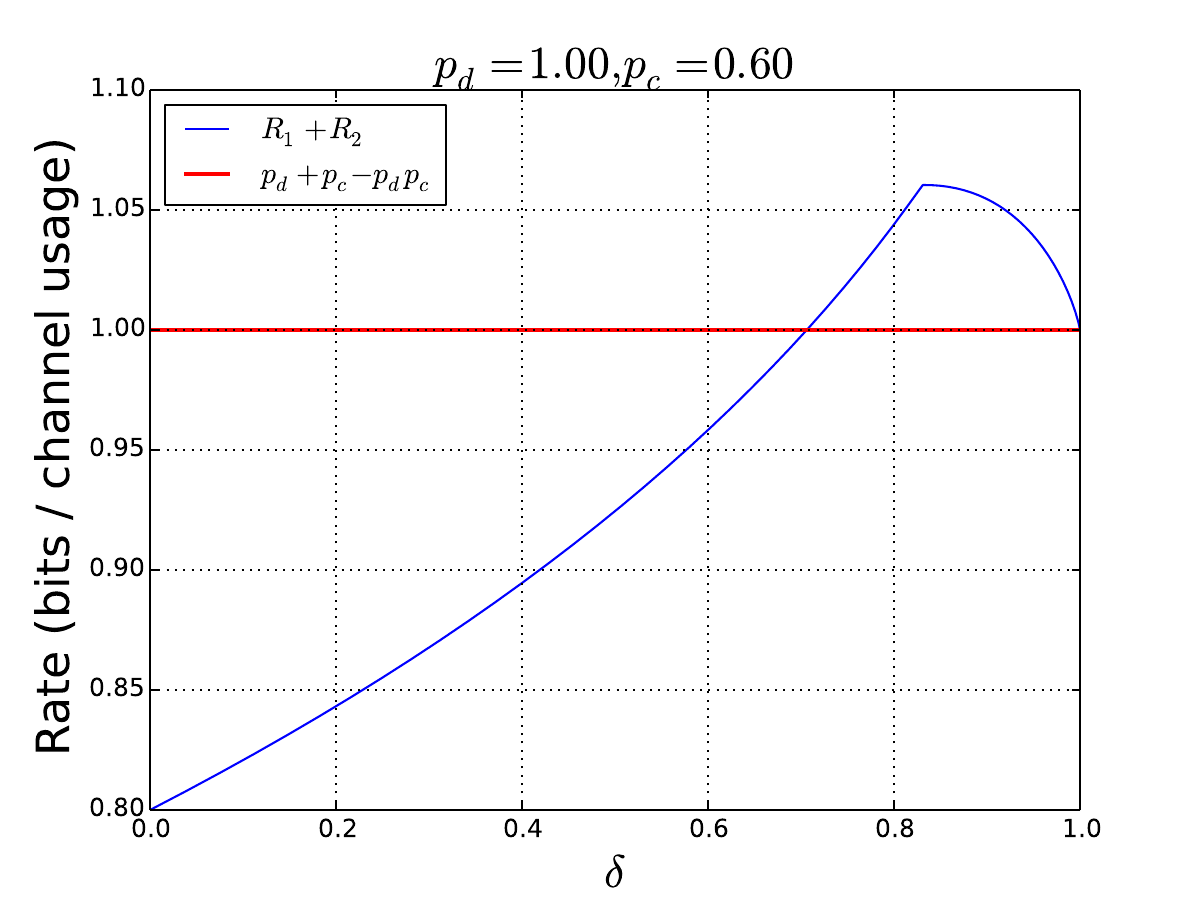}}
  \subfigure[]{\includegraphics[width=.5\textwidth]{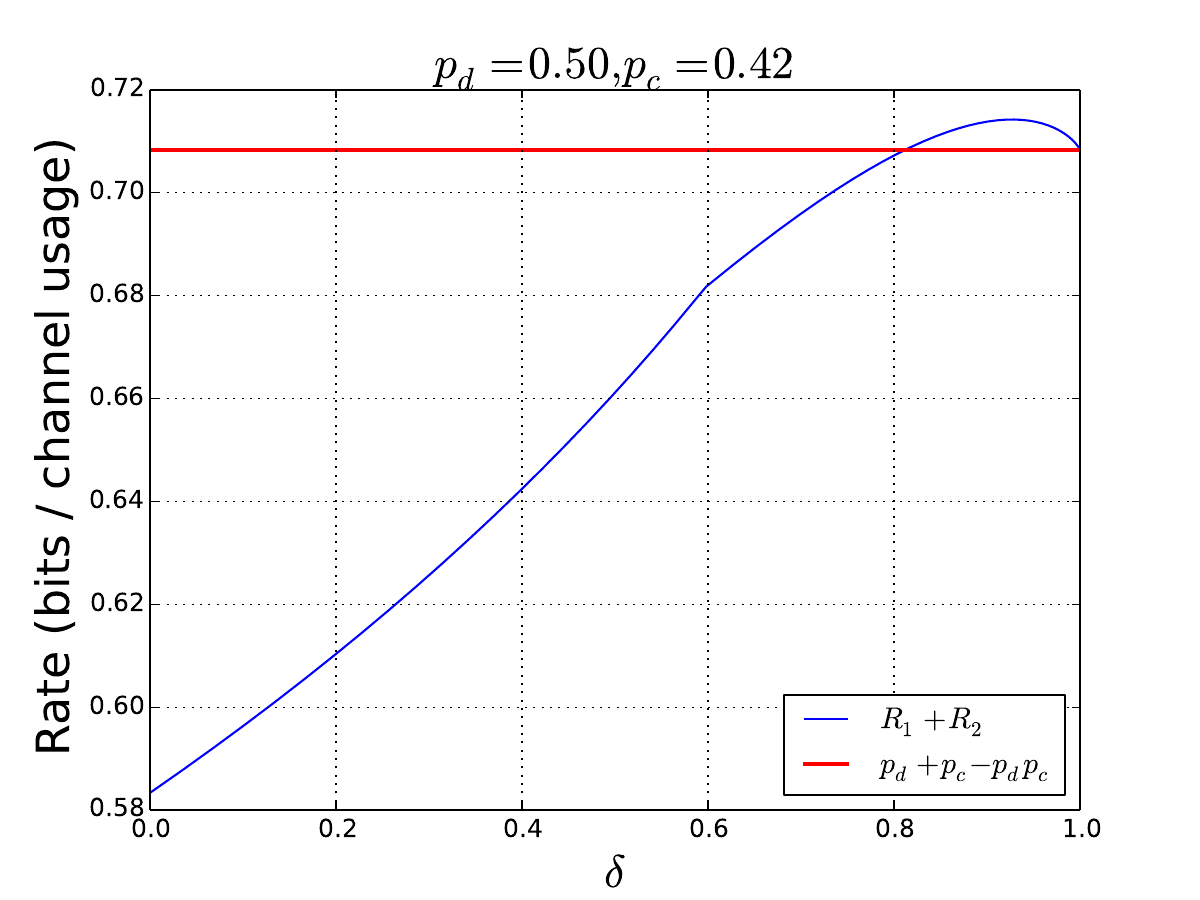}}
  \caption{Numerical Results for the achievable HK Rate}
  \label{fig:1}
\end{figure}

\noindent For the virtual MAC at receiver 2:
\begin{align}
& \CMI{\rX_{1c}, \rX_{2c}; \rY_2}{\rQ, \rG}  \nonumber \\
& = \CENT{\rY_2}{\rQ, \rG}  - \CENT{\rY_2}{\rX_{1c}, \rX_{2c}, \rQ, \rG} \NONUM \\
& = p_d + p_c - p_d p_c  - \lambda_1\CENT{\rY_2}{\rX_{1c}, \rX_{2c}, \rQ=1, \rG} \NONUM \\
&~- \lambda_2\CENT{\rY_2}{\rX_{1c}, \rX_{2c}, \rQ=2, \rG} \NONUM \\
& = p_d + p_c - p_d p_c - \lambda_1\Cdel{1} p_c - \lambda_2 \Cdel{2} p_d 
\end{align}
\begin{align}
& \CMI{ \rX_{2c}; \rY_2}{\rX_{1c}, \rQ, \rG}  \nonumber \\  
& = \CMI{\rX_{1c}, \rX_{2c}; \rY_2}{\rQ, \rG} - \CMI{\rX_{1c} ; \rY_2}{\rQ, \rG} \NONUM \\
& = p_d + p_c - p_d p_c - \lambda_1\Cdel{1} p_c - \lambda_2 \Cdel{2} p_d \NONUM \\
&~- \lambda_1p_c(1-p_d)(1-\Cdel{1}) - \lambda_2p_c(1-p_d) \NONUM \\
& = p_d - \lambda_1\Cdel{1} p_c - \lambda_2 \Cdel{2} p_d + \lambda_1p_c(1-p_d)\Cdel{1} \NONUM \\
& = p_d - \lambda_2 \Cdel{2} p_d - \lambda_1\Cdel{1} p_c p_d 
\end{align}
\begin{align}
& \CMI{ \rX_{1c}; \rY_2}{\rX_{2c}, \rQ, \rG}  \nonumber \\  
& = \CMI{\rX_{1c}, \rX_{2c}; \rY_2}{\rQ, \rG} - \CMI{\rX_{2c} ; \rY_2}{\rQ, \rG} \NONUM \\
& = p_d + p_c - p_d p_c - \lambda_1\Cdel{1} p_c - \lambda_2 \Cdel{2} p_d - \lambda_1p_d(1-p_c) \NONUM \\
&~- \lambda_2p_d(1-p_c)(1-\Cdel{2}) \NONUM \\
& = p_c - \lambda_1\Cdel{1} p_c - \lambda_2 \Cdel{2} p_d + \lambda_2p_d(1-p_c)\Cdel{2} \NONUM \\
& = p_c - \lambda_1 \Cdel{1} p_c - \lambda_2\Cdel{2} p_c p_d 
\end{align}

We reduce the number of constraints by performing the following comparisons.

\noindent $\bullet$ Comparing the two dominant faces:
\begin{align}
& \CMI{\rX_{1c}, \rX_{2c}; \rY_1}{\rQ, \rG} - \CMI{\rX_{1c}, \rX_{2c}; \rY_2}{\rQ, \rG}  \nonumber \\
& = p_d + p_c - p_d p_c - \lambda_1 \Cdel{1} p_d - \lambda_2 \Cdel{2} p_c  \NONUM\\ 
&~- \Big( p_d + p_c - p_d p_c - \lambda_1\Cdel{1} p_c - \lambda_2 \Cdel{2} p_d  \Big) \NONUM \\
& = (p_c - p_d) (\lambda_1\Cdel{1} - \lambda_2 \Cdel{2}).
\end{align}
Thus the sign is determined by $\left( \lambda_1\Cdel{1} - \lambda_2 \Cdel{2} \right)$. 

\noindent $\bullet$ Comparing single-user bounds for user 1's common message:
\begin{align}
& \CMI{\rX_{1c}; \rY_1}{\rX_{2c}, \rQ, \rG}  - \CMI{ \rX_{1c}; \rY_2}{\rX_{2c}, \rQ, \rG} \nonumber \\
& = p_d - \lambda_1 \Cdel{1}  p_d  - \lambda_2 \Cdel{2} p_cp_d  \NONUM \\
&~- \Big( p_c - \lambda_1 \Cdel{1} p_c - \lambda_2\Cdel{2} p_c p_d \Big) \NONUM \\
& = (p_d - p_c)(1 - \lambda_1 \Cdel{1}) \geq 0.
\end{align}
So only $\CMI{ \rX_{1c}; \rY_2}{\rX_{2c}, \rQ, \rG}$ matters. 

\noindent $\bullet$ Comparing single-user bounds for user 2's common message:
\begin{align}
& \CMI{\rX_{2c}; \rY_1}{\rX_{1c}, \rQ, \rG}  - \CMI{ \rX_{2c}; \rY_2}{\rX_{1c}, \rQ, \rG} \nonumber \\
& = p_c - \lambda_2 \Cdel{2} p_c - \lambda_1 \Cdel{1} p_d p_c \NONUM \\
&~- \Big( p_d - \lambda_2 \Cdel{2} p_d - \lambda_1\Cdel{1} p_c p_d  \Big) \NONUM \\
& = (p_c - p_d)(1 - \lambda_2\Cdel{2}) \leq 0. 
\end{align}
So only $\CMI{\rX_{2c}; \rY_1}{\rX_{1c}, \rQ, \rG}$ matters. 

In summary, the following rates of common messages are achievable:
\begin{subequations}
\begin{align}
  R_{1c} &\leq p_c - \lambda_1 \Cdel{1} p_c - \lambda_2\Cdel{2} p_c p_d  \\
  R_{2c} & \leq p_c - \lambda_2 \Cdel{2} p_c - \lambda_1 \Cdel{1} p_d p_c  \\
R_{1c} + R_{2c} & \leq   p_d + p_c - p_d p_c  - \\
&~\max(\lambda_1 \Cdel{1} p_d + \lambda_2 \Cdel{2} p_c, \lambda_1\Cdel{1} p_c + \lambda_2 \Cdel{2} p_d  ) \NONUM
\end{align}
\end{subequations}

First, we evaluate cases where $p_d-p_c \leq p_d p_c$. In particular, Fig.~\ref{fig:1} shows the achievable rate vs
splitting parameter $\delta$ for $(p_d, p_c)
= (1, 0.6)$ and $(1/2, 5/12)$.